\documentclass[11pt]{article}%
\usepackage{amsfonts}
\usepackage{ulem}
\usepackage{bm}
\usepackage{amsmath,amssymb,amsthm,enumerate,epsfig,graphicx,natbib}
\usepackage{ifthen,latexsym,syntonly}
\usepackage{natbib}
\usepackage{rotating}
\usepackage{lscape}
\usepackage{amsmath}
\usepackage{amssymb}
\usepackage{graphicx}
\usepackage{algorithm,caption}
\usepackage{algorithmic}
\usepackage{color}
\usepackage{extarrows}
\usepackage{float}
\usepackage[bottom]{footmisc}%
\providecommand{\U}[1]{\protect\rule{.1in}{.1in}}
\algsetup{indent=2em}

\newtheorem{theorem}{Theorem}

\newtheorem{proposition}{Proposition}
\newtheorem{assumption}{Assumption}
\theoremstyle{definition}

\newtheorem{remark}{Remark}
\def\baselinestretch{1.3}

\allowdisplaybreaks
\begin{document}

\title{Asymptotic expansion for the transition densities of stochastic differential
equations driven by the gamma processes}
\author{Fan Jiang\thanks{School of Mathematical Sciences, Peking University, China,
Email: jiangf\underline{\hbox to 0.15cm{}}math@pku.edu.cn}
\and Xin Zang\thanks{Corresponding author. School of Mathematical Sciences, Peking University, China,
Email: xzang@pku.edu.cn}
\and Jingping Yang\thanks{School of Mathematical Sciences, Peking University,
China, Email: yangjp@math.pku.edu.cn} }
\date{\today }
\maketitle

\begin{abstract}
In this paper, enlightened by the asymptotic expansion methodology developed
by \cite{Chenxu_Li_2013_AOS} and \cite{li2016estimating}, we propose a
Taylor-type approximation for the transition densities of the stochastic
differential equations (SDEs) driven by the gamma processes, a special type of
L\'{e}vy processes. After representing the transition density as a conditional
expectation of Dirac delta function acting on the solution of the related SDE,
the key technical method for calculating the expectation of multiple
stochastic integrals conditional on the gamma process is presented. To
numerically test the efficiency of our method, we examine the pure jump
Ornstein--Uhlenbeck (OU) model and its extensions to two jump-diffusion
models. For each model, the maximum relative error between our approximated
transition density and the benchmark density obtained by the inverse Fourier
transform of the characteristic function is sufficiently small, which shows
the efficiency of our approximated method.

\end{abstract}

\section{Introduction}

It is known that L\'{e}vy-driven stochastic differential equations (SDEs) have
been discussed in detail
\citep{applebaum2009levy, kunita2019stochastic,kohatsujump}. The
jump-diffusion SDE driven by the gamma process, as one important type of the
L\'{e}vy-driven SDEs, has been widely used in finance. For instance, the
Ornstein--Uhlenbeck (OU) type SDEs driven by the gamma processes were applied
for modeling the short rate \citep{eberlein2013simple} and the returns of S\&P
500 index \citep{james2018stochastic}. The various sensitivity indices for the
asset price dynamics driven by the gamma processes were discussed in
\cite{kawai2010sensitivity} and \cite{kawai2011greeks}. Note that the gamma
process is a pure-jump increasing L\'{e}vy process
\citep{yor2007some,ContTankov04,applebaum2009levy}. Starting from the gamma
process, the variance gamma process was defined \citep{madan1990variance, Madan_Carr_Chang_1998_EFR}.

For the financial applications mentioned above, the transition densities of
the related SDEs play a vital role \citep{barndorff2012levy,schoutens03}.
However, except for some special cases, the transition densities or even
characteristic functions of the SDEs usually do not admit closed-form
formulas, which brings difficulties for related applications. In this paper,
enlightened by the asymptotic expansion method presented in
\cite{Chenxu_Li_2013_AOS} and \cite{li2016estimating}, we propose a
Taylor-type closed-form expansion for the transition density of the
jump-diffusion SDE driven by the gamma process.

For the jump-diffusion SDE driven by the gamma process, we start from
representing its transition density as a conditional expectation of a Dirac
delta function acting on the solution of the related SDE, by applying the
theory of Malliavin calculus
\citep{Kanwal2004generalized,hayashi2008asymptotic,ishikawa2013stochastic,kunita2019stochastic}.
The main challenge in our method is to calculate the expectation of the
product of the values of a gamma process at different intermediate times,
conditional on the value of this gamma process at the terminal time.
Consequently, through the distributional property of gamma bridge discussed in
\cite{ribeiro2004valuing}, we express this type of conditional expectation as
a polynomial function of the value of this gamma process at the terminal time.
In this context, the expansion term of the transition density for any finite
order can be analytically calculated in an efficient manner.

To illustrate the efficiency of our method, we conduct numerical analyses
through three examples of the SDEs driven by the gamma processes, i.e., the
pure jump OU model, along with its extensions to the constant diffusion and
the square-root diffusion models. For each model, we compare the true
transition density obtained by the inverse Fourier transform of its
characteristic function with the approximated density obtained by our proposed
asymptotic expansion method. The numerical results show that our approximated
transition density can be efficiently calculated and converge rapidly to the
true density.

The rest of this paper is organized as follows. Section \ref{sec: model_setup}
lays our model setup and gives the general expression of the asymptotic
expansion. Section \ref{sec:calculate_K} provides detailed procedures for
explicitly representing the expansion terms. Section \ref{sec:numerical}
exhibits the numerical performance of our expansion method through three
concrete examples. Section \ref{sec:conclusion} concludes the paper.

\section{The model setup and approximation methodology}

\label{sec: model_setup}

\subsection{Preliminaries of Dirac delta function}

\label{subsec:dirac_delta}

Before introducing our asymptotic expansion methodology, we first give a brief
introduction of the Dirac delta function. Please refer to
\cite{Kanwal2004generalized}, \cite{hayashi2008asymptotic},
\cite{ishikawa2013stochastic} and \cite{kunita2019stochastic} for more details.

For ease of exposition later, we introduce the following notations and
concepts.\ Denote by $\mathcal{S}^{\prime}(\mathbb{R})$ the set of all
real-valued tempered distributions. According to Section 6.2 in
\cite{Kanwal2004generalized}, the Dirac delta function denoted as
$\delta\left(  \cdot\right)  $ and its associated derivative operators
$\frac{d^{\ell}\delta\left(  \cdot\right)  }{dx^{\ell}}\ $for $\ell\geq
1$\ belong to $\mathcal{S}^{\prime}(\mathbb{R})$. Here, for each $\ell\geq1$,
the derivative operator $\frac{d^{\ell}\delta\left(  \cdot\right)  }{dx^{\ell
}}$ is defined through an inner product with an indefinitely differentiable
function $f(\cdot)$ with compact support on $\mathbb{R}$, i.e.,
\begin{equation}
\left\langle \frac{d^{\ell}\delta(x-y)}{dx^{\ell}},f(x)\right\rangle
_{x}=(-1)^{\ell}\left\langle \delta(x-y),\frac{d^{\ell}f}{dx^{\ell}%
}(x)\right\rangle _{x} \label{integrate-by-part}%
\end{equation}
for a fixed $y\in\mathbb{R}$, where the inner product $\left\langle f\left(
x\right)  ,g\left(  x\right)  \right\rangle _{x}:=\int_{\mathbb{-\infty}%
}^{\infty}f(x)g(x)dx$, see Section 2.6 in \cite{Kanwal2004generalized} for
more details.

Let $D_{\infty}(\mathbb{R})$ be the set of all real-valued smooth
Wiener-Poisson functionals and $D_{\infty}^{\prime}(\mathbb{R})$ be the set of
all real-valued generalized Wiener-Poisson functionals
\citep{kunita2019stochastic}. According to Theorem 5.12.1 and equation (5.175)
in \cite{kunita2019stochastic}, for a tempered distribution $\Phi
\in\mathcal{S}^{\prime}(\mathbb{R})$, a regular nondegenerate Wiener-Poisson
functional $F\in D_{\infty}(\mathbb{R})$ and a smooth Wiener-Poisson
functional $G\in D_{\infty}(\mathbb{R})$, the generalized expectation
$\mathbb{E}\left[  \cdot\right]  $ is defined as
\begin{equation}
\mathbb{E}\left[  \Phi\left(  F\right)  \cdot G\right]  :=\frac{1}{2\pi}%
\int_{-\infty}^{+\infty}\int_{\mathbb{-\infty}}^{+\infty}e^{-ivx}\Phi\left(
x\right)  E\left[  Ge^{ivF}\right]  dxdv, \label{generalised Esp}%
\end{equation}
where $F$ and $G$ can be treated as the random variables on the Wiener-Poisson
space and the expectation in the right-hand side of (\ref{generalised Esp}) is
the usual expectation in common sense. Hereafter, the notations $\mathbb{E}%
\left[  \cdot\right]  $ and $E\left[  \cdot\right]  $ represent the
generalized expectation and usual expectation respectively.

For a fixed $y\in\mathbb{R}$, when we take $\Phi(\cdot)=\delta(\cdot-y)$ and
$G\equiv1$ in (\ref{generalised Esp}), from the equation%
\begin{equation}
\int_{\mathbb{-\infty}}^{+\infty}e^{-ivx}\delta(x-y)dx=e^{-ivy},
\label{int2_1}%
\end{equation}
we obtain that
\begin{equation}
\mathbb{E}\left[  \delta\left(  F-y\right)  \right]  =\frac{1}{2\pi}%
\int_{\mathbb{-\infty}}^{\infty}e^{-ivy}E\left[  e^{ivF}\right]  dv.
\label{delta_Fourier}%
\end{equation}
Moreover, the regular nondegenerate Wiener-Poisson functional $F$\ in
(\ref{delta_Fourier}) can also be taken as the strong solution of a
homogeneous jump-diffusion SDE satisfying the nondegenerate bounded (NDB)
condition, see Sections 3.5 -- 3.6 in \cite{ishikawa2013stochastic} for more
details. Especially, taking $F\equiv0$ in (\ref{delta_Fourier}), we can obtain
that%
\begin{equation}
\delta\left(  y\right)  =\frac{1}{2\pi}\int_{-\infty}^{+\infty}e^{ivy}dv.
\label{delta_exp}%
\end{equation}

Given a parameter $\epsilon>0$, let $F(\epsilon)$ be a regular Wiener-Poisson
functional in $D_{\infty}(\mathbb{R})$. For example, $F\left(  \epsilon
\right)  $ can be taken as the strong solution of some jump-diffusion SDE with
a parameter $\epsilon>0$. In the remaining of this section, for a fixed
$y\in\mathbb{R}$, we introduce some theoretical results about the asymptotic
expansion of $\delta(F(\epsilon)-y)$ with respect to the parameter $\epsilon$
(see Section 4.1 of \cite{ishikawa2013stochastic}). We assume that the
functional $F(\epsilon)$ satisfies the uniformly nondegenerate condition
(Definition 4.1 in \cite{ishikawa2013stochastic}), and has an expansion
\begin{equation}
F\left(  \epsilon\right)  =\sum_{j=0}^{\infty}f_{j}\epsilon^{j}
\label{expan_F_eps}%
\end{equation}
with respect to the norm in $D_{\infty}(\mathbb{R})$, where $f_{0},f_{1}%
,f_{2},\ldots$ are smooth Wiener-Poisson functionals.

According to Theorem 4.1 of \cite{ishikawa2013stochastic}, for each fixed
$y\in\mathbb{R}$, $\delta\left(  F\left(  \epsilon\right)  -y\right)  $
belongs to $D_{\infty}^{\prime}(\mathbb{R})$ and has an asymptotic expansion%
\begin{equation}
\delta\left(  F\left(  \epsilon\right)  -y\right)  =\sum_{m=0}^{M}\Phi
_{m}(y)\epsilon^{m}+\mathcal{O}(\epsilon^{M+1}) \label{expan_delta_F}%
\end{equation}
with respect to the norm in $D_{\infty}^{\prime}(\mathbb{R})$, where
$M\in\mathbb{N}$ denotes an arbitrary order of the expansion. Given the
functionals $\left\{  f_{0},f_{1},f_{2},\ldots\right\}  $ defined by
(\ref{expan_F_eps}), the coefficients $\Phi_{m}(y)\in D_{\infty}^{\prime
}(\mathbb{R})$ for $m\geq0$ can be expressed as
\begin{equation}
\Phi_{0}(y)=\delta\left(  f_{0}-y\right)  \text{ and }\Phi_{m}(y)=\sum
_{\left(  \ell,\left(  j_{1},j_{2},\cdots,j_{\ell}\right)  \right)
\in\mathcal{S}_{m}}\frac{1}{\ell!}\frac{d^{\ell}\delta\left(  f_{0}-y\right)
}{dx^{\ell}}\prod\limits_{i=1}^{\ell}f_{j_{i}}\text{ for }m\geq1,
\label{phi_general_def}%
\end{equation}
where the index set $\mathcal{S}_{m}$ is defined as
\begin{align}
\mathcal{S}_{m}  &  :=\left\{  \left.  \left(  \ell,\mathbf{j}\left(
\ell\right)  \right)  \right\vert \ell=1,2,\ldots,\text{ }\mathbf{j}\left(
\ell\right)  =\left(  j_{1},j_{2},\ldots,j_{\ell}\right)  \text{ with }%
j_{1},j_{2},\ldots,j_{\ell}\geq1\right. \nonumber\\
&  \text{ \ \ \ \ }\left.  \text{and }j_{1}+j_{2}+\cdots+j_{\ell}=m\right\}  .
\label{Sm}%
\end{align}
For example, the coefficients $\Phi_{1}(y)$ and $\Phi_{2}(y)$ are given by%
\[
\Phi_{1}(y)=f_{1}\frac{d\delta\left(  f_{0}-y\right)  }{dx}\text{ and }%
\Phi_{2}(y)=f_{2}\frac{d\delta\left(  f_{0}-y\right)  }{dx}+\frac{1}{2}%
f_{1}^{2}\frac{d^{2}\delta\left(  f_{0}-y\right)  }{dx^{2}}.
\]

According to Section 4 in \cite{ishikawa2013stochastic}, by taking the
generalized expectation defined in (\ref{generalised Esp}) on both sides of
equation (\ref{expan_delta_F}), we can obtain that
\begin{equation}
\mathbb{E}\left[  \delta\left(  F\left(  \epsilon\right)  -y\right)  \right]
=\sum_{m=0}^{M}\mathbb{E}\left[  \Phi_{m}(y)\right]  \epsilon^{m}%
+\mathcal{O}(\epsilon^{M+1}). \label{Esp_delta_F}%
\end{equation}
For example, the terms $\mathbb{E}\left[  \Phi_{0}(y)\right]  $ and
$\mathbb{E}\left[  \Phi_{1}(y)\right]  $ in (\ref{Esp_delta_F}) can be
evaluated via (\ref{integrate-by-part}), (\ref{generalised Esp}) and
(\ref{int2_1}) as%
\[
\mathbb{E}\left[  \Phi_{0}(y)\right]  =\mathbb{E}[\delta\left(  f_{0}%
-y\right)  ]=\frac{1}{2\pi}\int_{-\infty}^{+\infty}e^{-ivy}E[e^{ivf_{0}}]dv
\]
and%
\[
\mathbb{E}\left[  \Phi_{1}(y)\right]  =\mathbb{E}[f_{1}\frac{d\delta\left(
f_{0}-y\right)  }{dx}]=\frac{1}{2\pi}\int_{-\infty}^{+\infty}ive^{-ivy}%
E[f_{1}e^{ivf_{0}}]dv.
\]

In the subsequent calculation of the expansion terms for the transition
density, by using equation (\ref{generalised Esp}), we will transform some
specific generalized expectations like $\mathbb{E}\left[  \Phi_{m}(y)\right]
$ into the usual expectations.

\subsection{The model setup}

In this paper, we consider the following homogeneous jump-diffusion SDE driven
by a gamma process
\begin{equation}
dX(t)=\mu(X(t);\bm{\theta})dt+\sigma(X(t);\bm{\theta})dW(t)+dL(t),\text{
}X\left(  0\right)  =x_{0}, \label{model}%
\end{equation}
where the functions $\mu\left(  x;\bm{\theta}\right)  $ and $\sigma
(x;\bm{\theta})$\ are assumed to depend on some parameter vector $\bm{\theta}$
belonging to an open bound set $\Theta$, $\left\{  W(t),t\geq0\right\}  $ is a
Brownian motion and $\left\{  L(t),t\geq0\right\}  $ is a gamma process.
Moreover, we assume that the gamma process $\left\{  L(t),t\geq0\right\}  $
starts at $L(0)=0$ with the density function%
\begin{equation}
p_{L\left(  t\right)  }\left(  x\right)  =\frac{b^{at}x^{at-1}e^{-bx}}%
{\Gamma\left(  at\right)  },\text{ }x\geq0 \label{density_gamma process}%
\end{equation}
at time $t>0$, where $a$ and $b$ are positive constants, and $\Gamma(\cdot)$
denotes the gamma function. The two processes $\left\{  W(t),t\geq0\right\}  $
and $\left\{  L(t),t\geq0\right\}  $ are independent. The characteristic
function of the gamma process $L(t)$ is calculated as%
\begin{equation}
Ee^{i\lambda L(t)}=\left(  1-\frac{i\lambda}{b}\right)  ^{-at}\triangleq
\exp\left[  t\psi_{L}\left(  \lambda\right)  \right]  ,
\label{characteristic exponent}%
\end{equation}
where $\psi_{L}\left(  \lambda\right)  =-a\log\left(  1-i\lambda/b\right)  $
is the characteristic exponent of $L(t)$.

To guarantee the existence and uniqueness of the strong solution $X(t)$\ of
SDE (\ref{model})\ and obtain other desirable properties for implementing our
method, the following standard and technical assumptions are assumed in this paper:

\begin{assumption}
The diffusion function $\sigma\left(  x;\bm{\theta}\right)  $ satisfies that
$\inf_{x\in\mathbb{R}}\sigma\left(  x;\bm{\theta}\right)  >0$\ for any
$\bm{\theta}\in\Theta$. \label{asmp_diffusion}
\end{assumption}

\begin{assumption}
For each $k\in\mathbb{N}_{+}$, the $k$-th order partial derivatives in $x$ of
$\mu\left(  x;\bm{\theta}\right)  $ and $\sigma(x;\bm{\theta})$ are uniformly
bounded for any $\left(  x,\bm{\theta}\right)  \in\mathbb{R}\times\Theta$.
\label{asmp_derivatives}
\end{assumption}

\begin{assumption}
The functions $\mu\left(  x;\bm{\theta}\right)  $ and $\sigma(x;\bm{\theta})$
satisfy the linear growth conditions%
\[
\left\vert \mu\left(  x;\bm{\theta}\right)  \right\vert \leq c_{1}\left(
1+\left\vert x\right\vert \right)  \text{ and }\left\vert \sigma\left(
x;\bm{\theta}\right)  \right\vert \leq c_{2}\left(  1+\left\vert x\right\vert
\right)  ,
\]
for some $c_{1},c_{2}\in\mathbb{R}_{+}$ and any $\left(  x,\bm{\theta}\right)
\in\mathbb{R}\times\Theta$.\label{asmp_linear_growth}
\end{assumption}

Assumption \ref{asmp_diffusion} and Assumption \ref{asmp_derivatives}
guarantee the NDB condition and uniformly nondegenerate condition for
justifying the validity and convergence of our proposed asymptotic expansion
method, which will be shown in section \ref{subsec: expan_tran_density}.
Assumption \ref{asmp_derivatives} and Assumption \ref{asmp_linear_growth}
guarantee the existence and uniqueness of the strong solution $X(t)$\ of SDE
(\ref{model}).

\subsection{The expansion of the transition density}

\label{subsec: expan_tran_density}

For the jump-diffusion SDE (\ref{model}), by the time-homogeneity nature, the
transition density of $X\left(  t+\Delta\right)  $ given $X\left(  t\right)
=x_{0}$ can be expressed as
\begin{equation}
\mathbb{P}\left(  \left.  X\left(  t+\Delta\right)  \in dx\right\vert X\left(
t\right)  =x_{0};\bm{\theta}\right)  =p_{X\left(  \Delta\right)  }\left(
x|x_{0};\bm{\theta}\right)  dx, \label{transition density}%
\end{equation}
where $\Delta$ denotes the time interval. For most SDEs defined in
(\ref{model}), their transition densities do not admit closed-form
expressions. Even for some special cases with closed-form conditional
characteristic functions, the inversion to transition densities may not be
easy, especially for the pure jump processes
\citep{barndorff2012levy, schoutens03}. In the following, we propose a
closed-form expansion for approximating the transition density $p_{X\left(
\Delta\right)  }\left(  x|x_{0};\bm{\theta}\right)  $ of the jump-diffusion
SDE (\ref{model}).

To start with, we parameterize the dynamics of $X(t)$ in (\ref{model}) via a
parameter $\epsilon\in\left(  0,1\right]  $ as
\begin{equation}
dX(\epsilon,t)=\epsilon\left[  \mu(X(\epsilon,t);\bm{\theta})dt+\sigma\left(
X(\epsilon,t);\bm{\theta}\right)  dW(t)+dL(t)\right]  ,\text{ }X(\epsilon
,0)=x_{0}. \label{S-L basic model}%
\end{equation}
Note that the solution $X(\epsilon,t)$ of (\ref{S-L basic model}) satisfies
$\left.  X(\epsilon,t)\right\vert _{\epsilon=1}=X(t)$. By regarding
$\epsilon\in\left(  0,1\right]  $ as an extra element of the parameter vector,
we see that the SDE (\ref{S-L basic model}) still satisfies Assumption
\ref{asmp_derivatives} and Assumption \ref{asmp_linear_growth}, which implies
the existence and uniqueness of the strong solution $X(\epsilon,t)$
\citep{platen2010numerical}. The transition density of $X(\epsilon,t)$ in
(\ref{S-L basic model})\ can be expressed as
\begin{equation}
\mathbb{P}\left(  \left.  X(\epsilon,t+\Delta)\in dx\right\vert X\left(
\epsilon,t\right)  =x_{0};\bm{\theta}\right)  =p_{X\left(  \epsilon
,\Delta\right)  }\left(  x|x_{0};\bm{\theta}\right)  dx.
\label{X_eps_tran_density}%
\end{equation}
Once we obtain an asymptotic expansion of $p_{X\left(  \epsilon,\Delta\right)
}\left(  x|x_{0};\bm{\theta}\right)  $ as a series of $\epsilon$, the
transition density $p_{X\left(  \Delta\right)  }\left(  \left.  x\right\vert
x_{0};\bm{\theta}\right)  $ in (\ref{transition density})\ can be obtained by
letting $\epsilon=1$.

To derive the asymptotic expansion of $p_{X\left(  \epsilon,\Delta\right)
}\left(  x|x_{0};\bm{\theta}\right)  $ in (\ref{X_eps_tran_density}), we first
claim that $X(\epsilon,t)$ satisfies the NDB condition introduced in Section
\ref{subsec:dirac_delta}, which will justify the representation of the
transition density $p_{X\left(  \epsilon,\Delta\right)  }\left(
x|x_{0};\bm{\theta}\right)  $ as a conditional expectation shown below. Note
that for $X(\epsilon,t)$, the NDB condition introduced from Definition 3.5 in
\cite{ishikawa2013stochastic} is transformed to the condition that
$\sigma\left(  x;\bm{\theta}\right)  \neq0$ for any $\left(
x,\bm{\theta}\right)  \in\mathbb{R}\times\Theta$, which is guaranteed by
Assumption \ref{asmp_diffusion}.

Based on the NDB condition and the time-homogeneity nature of $X(\epsilon,t)$,
we represent\newline$p_{X\left(  \epsilon,\Delta\right)  }\left(
x|x_{0};\bm{\theta}\right)  $ as a conditional expectation of Dirac delta
function acting on $X\left(  \epsilon,\Delta\right)  -x$ by
\begin{equation}
p_{X\left(  \epsilon,\Delta\right)  }\left(  \left.  x\right\vert
x_{0};\bm{\theta}\right)  =\mathbb{E}\left[  \left.  \delta\left(  X\left(
\epsilon,\Delta\right)  -x\right)  \right\vert X(\epsilon,0)=x_{0}%
;\bm{\theta}\right]  . \label{p^eps density}%
\end{equation}
The validity of (\ref{p^eps density}) will be verified in Remark
\ref{remark:dirac delta} below in detail. For brevity, we omit the initial
condition $X(\epsilon,0)=x_{0}$ and drop the dependence of $\bm{\theta}$\ in
the dynamics of (\ref{model}) and (\ref{S-L basic model}) hereafter, unless
especially noted. The starting point of the expansion for (\ref{p^eps density}%
) lies in that $X\left(  \epsilon,\Delta\right)  $ admits the pathwise
Taylor-type expansion%
\begin{equation}
X\left(  \epsilon,\Delta\right)  =\sum_{m=0}^{M}X_{m}(\Delta)\epsilon
^{m}+\mathcal{O}(\epsilon^{M+1}), \label{X_path_expan}%
\end{equation}
where $M\in\mathbb{N}$ denotes an arbitrary order of expansion, see, e.g.,
Chapter 4 in \cite{platen2010numerical} for the validity of this expansion.

In the following, we first derive the expressions of the expansion terms
$X_{m}(\Delta)$ in (\ref{X_path_expan}). We rewrite (\ref{S-L basic model}) in
integrated form as%

\begin{equation}
X(\epsilon,\Delta)=x_{0}+\epsilon\left[  \int_{0}^{\Delta}\mu(X(\epsilon
,s))ds+\int_{0}^{\Delta}\sigma\left(  X(\epsilon,s)\right)  dW(s)+L(\Delta
)\right]  . \label{X_eps_integral}%
\end{equation}
Letting $\epsilon\downarrow0$ on both sides in (\ref{X_eps_integral}), we
have
\begin{equation}
X(0,\Delta)=x_{0} \label{X_0}%
\end{equation}
for any fixed $\Delta\geq0$. Further, from the expansion of $X\left(
\epsilon,\Delta\right)  $ in (\ref{X_path_expan}), we can obtain that
\begin{equation}
\mu\left(  X(\epsilon,t)\right)  :=\sum_{m=0}^{M}\mu_{m}(t)\epsilon
^{m}+\mathcal{O}(\epsilon^{M+1}) \label{mu_expan}%
\end{equation}
and%
\begin{equation}
\sigma\left(  X(\epsilon,t)\right)  :=\sum_{m=0}^{M}\sigma_{m}(t)\epsilon
^{m}+\mathcal{O}(\epsilon^{M+1}) \label{sigma_expan}%
\end{equation}
in the SDE (\ref{S-L basic model}), where
\begin{equation}
\mu_{m}(t):=\frac{1}{m!}\left.  \frac{d^{m}\mu\left(  X(\epsilon,t)\right)
}{d\epsilon^{m}}\right\vert _{\epsilon=0}=\sum_{\left(  \ell,\left(
j_{1},j_{2},\cdots,j_{\ell}\right)  \right)  \in\mathcal{S}_{m}}\frac{1}%
{\ell!}\frac{d^{\ell}\mu\left(  x_{0}\right)  }{dx^{\ell}}\prod\limits_{i=1}%
^{\ell}X_{j_{i}}\left(  t\right)  \label{mu_m(t)}%
\end{equation}
and%
\begin{equation}
\sigma_{m}(t):=\frac{1}{m!}\left.  \frac{d^{m}\sigma\left(  X(\epsilon
,t)\right)  }{d\epsilon^{m}}\right\vert _{\epsilon=0}=\sum_{\left(
\ell,\left(  j_{1},j_{2},\cdots,j_{\ell}\right)  \right)  \in\mathcal{S}_{m}%
}\frac{1}{\ell!}\frac{d^{\ell}\sigma\left(  x_{0}\right)  }{dx^{\ell}}%
\prod\limits_{i=1}^{\ell}X_{j_{i}}\left(  t\right)  \label{sigmam(t)}%
\end{equation}
with the index set $\mathcal{S}_{m}$ defined by (\ref{Sm}) and the condition
$X(0,t)=x_{0}$ as in (\ref{X_0}). Plugging (\ref{X_path_expan}),
(\ref{mu_expan}), and (\ref{sigma_expan}) into (\ref{X_eps_integral}), we have%
\begin{align*}
&  \sum_{m=0}^{M}X_{m}(\Delta)\epsilon^{m}+\mathcal{O}(\epsilon^{M+1})\\
=  &  x_{0}+\sum_{m=0}^{M}\left(  \int_{0}^{\Delta}\mu_{m}(s)ds+\int%
_{0}^{\Delta}\sigma_{m}(s)dW(s)\right)  \epsilon^{m+1}+\epsilon L(\Delta
)+\mathcal{O}(\epsilon^{M+2}).
\end{align*}
By comparing the coefficients of $\epsilon^{m}$ for $m=0,1,2,\ldots$, we
conclude that $X_{0}\left(  \Delta\right)  \equiv x_{0}$,
\begin{equation}
X_{1}\left(  \Delta\right)  =\mu\left(  x_{0}\right)  \Delta+\sigma
(x_{0})W(\Delta)+L\left(  \Delta\right)  , \label{X_1}%
\end{equation}
and
\begin{equation}
X_{m}(\Delta)=\int_{0}^{\Delta}\mu_{m-1}(s)ds+\int_{0}^{\Delta}\sigma
_{m-1}(s)dW(s)\text{ for }m\geq2, \label{X_m}%
\end{equation}
where $\mu_{m-1}(s)$ and $\sigma_{m-1}(s)$ are defined in equations
(\ref{mu_m(t)}) and (\ref{sigmam(t)}) respectively. As the expressions
$\mu_{m-1}(s)$ and $\sigma_{m-1}(s)$ involved in the right-hand side of
(\ref{X_m}) are determined by the expansion terms $X_{0}\left(  s\right)  ,$
$X_{1}(s),\ldots,X_{m-1}(s)$ in (\ref{X_path_expan}) with orders at most
$m-1$, we notice that $X_{m}(\Delta)$ is fully determined by $X_{i}%
(s),i=0,1,2,\ldots,m-1$ for all $0\leq s\leq\Delta$.

Next, we illustrate the expansion of $p_{X\left(  \epsilon,\Delta\right)
}\left(  x|x_{0};\bm{\theta}\right)  $ in (\ref{p^eps density}) as a
convergent series of $\epsilon$. To do this, we standardize $X(\epsilon
,\Delta)$ into%
\begin{equation}
Y(\epsilon,\Delta)=\frac{X(\epsilon,\Delta)-x_{0}}{\sigma(x_{0})\sqrt{\Delta
}\epsilon}, \label{Y_eps_delta}%
\end{equation}
from which the transition density $p_{X\left(  \epsilon,\Delta\right)
}\left(  x|x_{0};\bm{\theta}\right)  $ in (\ref{p^eps density}) can be
represented in terms of $Y\left(  \epsilon,\Delta\right)  $ as%
\begin{equation}
p_{X\left(  \epsilon,\Delta\right)  }\left(  x|x_{0};\bm{\theta}\right)
=\left.  \frac{1}{\sigma(x_{0})\sqrt{\Delta}\epsilon}\mathbb{E}\left[
\delta\left(  Y\left(  \epsilon,\Delta\right)  -y\right)  \right]  \right\vert
_{y=\frac{x-x_{0}}{\sigma(x_{0})\sqrt{\Delta}\epsilon}}. \label{density:X-Y}%
\end{equation}
According to \cite{hayashi2012composition} and Definition 4.1 in
\cite{ishikawa2013stochastic}, $Y(\epsilon,\Delta)$ in (\ref{Y_eps_delta}%
)\ satisfies the uniformly nondegenerate condition which is guaranteed by
Assumption \ref{asmp_diffusion} and Assumption \ref{asmp_derivatives}. Then
the expectation in the right-hand side of (\ref{density:X-Y}) admits a
convergent series of $\epsilon$. To obtain an expansion of $\mathbb{E}\left[
\delta\left(  Y\left(  \epsilon,\Delta\right)  -y\right)  \right]  $ in
(\ref{density:X-Y}) with respect to $\epsilon$, we notice from
(\ref{X_path_expan}), (\ref{Y_eps_delta}) and $X_{0}\left(  \Delta\right)
\equiv x_{0}$ that%
\begin{equation}
Y\left(  \epsilon,\Delta\right)  =\sum_{m=0}^{M}Y_{m}(\Delta)\epsilon
^{m}+\mathcal{O}(\epsilon^{M+1}), \label{Y_eps_delta_expan}%
\end{equation}
where%
\begin{equation}
Y_{m}(\Delta)=\frac{X_{m+1}(\Delta)}{\sigma(x_{0})\sqrt{\Delta}},\text{ for
}m=0,1,2,\ldots. \label{X-Y}%
\end{equation}
Since the functional $\delta\left(  \cdot-y\right)  $ belongs to
$\mathcal{S}^{\prime}(\mathbb{R})$, according to (\ref{expan_delta_F}), we
obtain a Taylor-type expansion of $\delta\left(  Y\left(  \epsilon
,\Delta\right)  -y\right)  $ as%
\begin{equation}
\delta\left(  Y\left(  \epsilon,\Delta\right)  -y\right)  =\sum_{m=0}^{M}%
\Phi_{m}(y)\epsilon^{m}+\mathcal{O}(\epsilon^{M+1}) \label{delta_Y_eps_expan}%
\end{equation}
for any $M\in\mathbb{N}$. Here in (\ref{delta_Y_eps_expan}), it follows from
(\ref{phi_general_def}), (\ref{Y_eps_delta_expan}) and (\ref{X-Y}) that%
\begin{equation}
\Phi_{0}\left(  y\right)  =\delta\left(  Y_{0}(\Delta)-y\right)  \label{phi0}%
\end{equation}
and%
\begin{equation}
\Phi_{m}(y)=\sum_{\left(  \ell,\left(  j_{1},j_{2},\cdots,j_{\ell}\right)
\right)  \in\mathcal{S}_{m}}\frac{1}{\ell!}\frac{1}{(\sigma(x_{0})\sqrt
{\Delta})^{\ell}}\frac{d^{\ell}\delta\left(  Y_{0}(\Delta)-y\right)
}{dx^{\ell}}\prod\limits_{i=1}^{\ell}X_{j_{i}+1}(\Delta) \label{phim}%
\end{equation}
for $m\geq1$, with the index set $\mathcal{S}_{m}$ defined in (\ref{Sm}) and
$Y_{0}(\Delta)=\frac{X_{1}(\Delta)}{\sigma(x_{0})\sqrt{\Delta}}$. Further,
based on (\ref{Esp_delta_F}), we take the generalized expectation on both
sides of (\ref{delta_Y_eps_expan}) to obtain that%
\begin{equation}
\mathbb{E}\left[  \delta\left(  Y\left(  \epsilon,\Delta\right)  -y\right)
\right]  =\sum_{m=0}^{M}\Omega_{m}(y)\epsilon^{m}+\mathcal{O}(\epsilon^{M+1}),
\label{Esp_delta_Y_eps}%
\end{equation}
where the generalized expectation $\Omega_{m}\left(  y\right)  :=\mathbb{E}%
\Phi_{m}\left(  y\right)  $ for $m\geq0$ will be explicitly derived and
transformed into some usual expectation below.

For simplicity, we name $\Omega_{0}\left(  y\right)  $ and $\Omega_{m}\left(
y\right)  $ for $m\geq1$ in (\ref{Esp_delta_Y_eps})\ the leading term and the
higher-order terms\ respectively. Combining (\ref{density:X-Y}) and
(\ref{Esp_delta_Y_eps}) by letting $\epsilon=1$, the approximated transition
density of $X(\Delta)$ up to the $M$-th order is proposed as
\begin{equation}
p_{X\left(  \Delta\right)  }^{\left(  M\right)  }\left(  \left.  x\right\vert
x_{0};\bm{\theta}\right)  :=\frac{1}{\sigma(x_{0})\sqrt{\Delta}}\sum_{m=0}%
^{M}\Omega_{m}\left(  \frac{x-x_{0}}{\sigma(x_{0})\sqrt{\Delta}}\right)  .
\label{pX_M_expan}%
\end{equation}
Consequently, to approximate the transition density $p_{X\left(
\Delta\right)  }\left(  \left.  x\right\vert x_{0};\bm{\theta}\right)  $ in
(\ref{transition density}) up to any finite order, it suffices to specify the
functions $\Omega_{0}\left(  y\right)  $ and $\Omega_{m}\left(  y\right)  $
for $m\geq1$ in (\ref{pX_M_expan}), which will be investigated in Section
\ref{subsec:omega_m}.

\begin{remark}
\label{remark:dirac delta} The equation (\ref{p^eps density}) can be verified
as follows. According to Section 6.4 in \cite{kunita2019stochastic}, under the
initial condition $X\left(  \epsilon,0\right)  =x_{0}$, the strong solution
$X\left(  \epsilon,t\right)  $ of SDE (\ref{S-L basic model}) is uniquely tied
to a stochastic flow, which is a regular Wiener-Poisson functional belonging
to the set $D_{\infty}(\mathbb{R})$ (Section 3.1 in
\cite{kunita2019stochastic}). Thus, by noting that $X(\epsilon,t)$ satisfies
the NDB condition, under the initial condition $X\left(  \epsilon,0\right)
=x_{0}$, we take $F=X\left(  \epsilon,\Delta\right)  $ in (\ref{delta_Fourier}%
) to obtain that%
\begin{align*}
&  \mathbb{E}\left[  \left.  \delta\left(  X\left(  \epsilon,\Delta\right)
-x\right)  \right\vert X(\epsilon,0)=x_{0};\bm{\theta}\right] \\
=  &  \frac{1}{2\pi}\int_{-\infty}^{+\infty}e^{-ivx}E\left[  \left.
e^{ivX\left(  \epsilon,\Delta\right)  }\right\vert X(\epsilon,0)=x_{0}%
;\bm{\theta}\right]  dv=p_{X\left(  \epsilon,\Delta\right)  }\left(
x|x_{0};\bm{\theta}\right)  ,
\end{align*}
which verifies (\ref{p^eps density}).
\end{remark}

\begin{remark}
\label{remark:PJ} The above closed-form expansion method can also be applied
to the special case of SDE (\ref{model}) with $\sigma(X(t);\bm{\theta})\equiv
0$. In this context, we only adjust the above algorithm to standardize
$X\left(  \epsilon,\Delta\right)  $ defined by (\ref{S-L basic model})\ into
\[
Y\left(  \epsilon,\Delta\right)  =\frac{X\left(  \epsilon,\Delta\right)
-x_{0}}{\epsilon}%
\]
instead of (\ref{Y_eps_delta}), which implies that%
\[
\mathbb{E}\left[  \delta\left(  X\left(  \epsilon,\Delta\right)  -x\right)
\right]  =\left.  \frac{1}{\epsilon}\mathbb{E}\left[  \delta\left(  Y\left(
\epsilon,\Delta\right)  -y\right)  \right]  \right\vert _{y=\frac{x-x_{0}%
}{\epsilon}}.
\]
The remaining procedures are performed in a similar manner. Therefore, the
approximated transition density of $X(\Delta)$ up to the $M$-th order can be
obtained by
\begin{equation}
p_{X\left(  \Delta\right)  }^{\left(  M\right)  }\left(  \left.  x\right\vert
x_{0};\bm{\theta}\right)  =\sum_{m=0}^{M}\Omega_{m}\left(  x-x_{0}\right)  ,
\label{pX_M_expan_PJ}%
\end{equation}
where $\Omega_{m}(y)=E\Phi_{m}\left(  y\right)  $ and the terms $\Phi
_{m}\left(  y\right)  $ are calculated by
\[
\Phi_{0}\left(  y\right)  =\delta\left(  Y_{0}(\Delta)-y\right)
\]
and%
\[
\Phi_{m}(y)=\sum_{\left(  \ell,\left(  j_{1},j_{2},\cdots,j_{\ell}\right)
\right)  \in\mathcal{S}_{m}}\frac{1}{\ell!}\frac{d^{\ell}\delta\left(
Y_{0}(\Delta)-y\right)  }{dx^{\ell}}\prod\limits_{i=1}^{\ell}X_{j_{i}%
+1}(\Delta)
\]
for $m\geq1$, with the index set $\mathcal{S}_{m}$ defined in (\ref{Sm}) and
$Y_{0}(\Delta)=X_{1}(\Delta)$.
\end{remark}

\subsection{General expressions of the leading term and high-order terms}

\label{subsec:omega_m}

In this part, we give the explicit expression of the leading term $\Omega
_{0}\left(  y\right)  $ and the general representations of the higher-order
terms $\Omega_{m}\left(  y\right)  $ for $m\geq1$ defined\ in
(\ref{Esp_delta_Y_eps}). Throughout this section, we denote by $\phi\left(
\cdot\right)  $ the density function of a standard normal variable and recall
that $p_{L\left(  t\right)  }\left(  \cdot\right)  $ is the density function
of the gamma process $L(t)$ given by (\ref{density_gamma process}).

From (\ref{phi0}), the leading term $\Omega_{0}\left(  y\right)  $ is
expressed as%
\begin{equation}
\Omega_{0}(y)=\mathbb{E}\left[  \delta\left(  Y_{0}(\Delta)-y\right)  \right]
, \label{omega_0(y)}%
\end{equation}
which is exactly the density function of $Y_{0}(\Delta)$ evaluated at $y$. The
explicit expression of $\Omega_{0}(y)$ is given in the following proposition.

\begin{proposition}
The leading term $\Omega_{0}\left(  y\right)  $ in (\ref{pX_M_expan}) admits
the following explicit expression
\[
\Omega_{0}(y)=\int_{0}^{+\infty}\phi\left(  y-\frac{\mu\left(  x_{0}\right)
\Delta+u}{\sigma(x_{0})\sqrt{\Delta}}\right)  p_{L\left(  \Delta\right)
}\left(  u\right)  du,
\]
where $p_{L\left(  \Delta\right)  }(\cdot)$ is the density function of the
gamma process $L(\Delta)$ given by (\ref{density_gamma process}).
\end{proposition}

\begin{proof}
From (\ref{delta_Fourier}) and (\ref{omega_0(y)}), we obtain that%
\begin{align}
\Omega_{0}(y)  &  =\mathbb{E}\left[  \delta\left(  Y_{0}(\Delta)-y\right)
\right] \nonumber\\
&  =\frac{1}{2\pi}\int_{-\infty}^{+\infty}e^{-ivy}E\left[  e^{ivY_{0}\left(
\Delta\right)  }\right]  dv\nonumber\\
&  =E\left[  \frac{1}{2\pi}\int_{-\infty}^{+\infty}e^{-ivy}E\left[  \left.
e^{ivY_{0}\left(  \Delta\right)  }\right\vert L\left(  \Delta\right)  \right]
dv\right]  . \label{int2_2}%
\end{align}
We notice from (\ref{X_1}) and (\ref{X-Y}) that $Y_{0}(\Delta)$ in
(\ref{int2_2}) can be represented as
\begin{equation}
Y_{0}(\Delta)=\frac{X_{1}(\Delta)}{\sigma(x_{0})\sqrt{\Delta}}=\frac
{W(\Delta)}{\sqrt{\Delta}}+\frac{\mu\left(  x_{0}\right)  \Delta+L\left(
\Delta\right)  }{\sigma(x_{0})\sqrt{\Delta}}. \label{Y_0}%
\end{equation}
Here, conditioned on the jump term $L\left(  \Delta\right)  $, the variable
$Y_{0}(\Delta)$ in (\ref{Y_0})\ follows a normal distribution. Therefore, the
inner term of the expectation in the last equation of (\ref{int2_2}) can be
calculated as%
\begin{align}
&  \frac{1}{2\pi}\int_{-\infty}^{+\infty}e^{-ivy}E\left[  \left.
e^{ivY_{0}\left(  \Delta\right)  }\right\vert L\left(  \Delta\right)  \right]
dv\nonumber\\
=  &  E\left[  \left.  \frac{1}{2\pi}\int_{-\infty}^{+\infty}e^{-ivy}%
e^{iv\left(  \frac{W(\Delta)}{\sqrt{\Delta}}+\frac{\mu\left(  x_{0}\right)
\Delta+L\left(  \Delta\right)  }{\sigma(x_{0})\sqrt{\Delta}}\right)
}dv\right\vert L\left(  \Delta\right)  \right] \nonumber\\
=  &  \int_{-\infty}^{\infty}E\left[  \left.  \frac{1}{2\pi}\int_{-\infty
}^{+\infty}e^{-ivy}e^{iv\left(  x+\frac{\mu\left(  x_{0}\right)
\Delta+L\left(  \Delta\right)  }{\sigma(x_{0})\sqrt{\Delta}}\right)
}dv\right\vert L\left(  \Delta\right)  \right]  \phi\left(  x\right)
dx\nonumber\\
=  &  E\left[  \left.  \int_{-\infty}^{\infty}\left(  \frac{1}{2\pi}%
\int_{-\infty}^{+\infty}e^{-ivy}e^{iv\left(  x+\frac{\mu\left(  x_{0}\right)
\Delta+L\left(  \Delta\right)  }{\sigma(x_{0})\sqrt{\Delta}}\right)
}dv\right)  \phi\left(  x\right)  dx\right\vert L\left(  \Delta\right)
\right]  . \label{int2_3}%
\end{align}
By the relation (\ref{delta_exp}), we obtain that
\[
\frac{1}{2\pi}\int_{-\infty}^{+\infty}e^{-ivy}e^{iv\left(  x+\frac{\mu\left(
x_{0}\right)  \Delta+L\left(  \Delta\right)  }{\sigma(x_{0})\sqrt{\Delta}%
}\right)  }dv=\delta\left(  x+\frac{\mu\left(  x_{0}\right)  \Delta+L\left(
\Delta\right)  }{\sigma(x_{0})\sqrt{\Delta}}-y\right)  .
\]
Then plugging the above equation into (\ref{int2_3}) by noting the definition
of the Dirac delta function, we obtain that%
\begin{align}
\frac{1}{2\pi}\int_{-\infty}^{+\infty}e^{-ivy}E\left[  \left.  e^{ivY_{0}%
\left(  \Delta\right)  }\right\vert L\left(  \Delta\right)  \right]  dv  &
=E\left[  \left.  \int_{-\infty}^{\infty}\delta\left(  x+\frac{\mu\left(
x_{0}\right)  \Delta+L\left(  \Delta\right)  }{\sigma(x_{0})\sqrt{\Delta}%
}-y\right)  \phi\left(  x\right)  dx\right\vert L\left(  \Delta\right)
\right] \nonumber\\
&  =E\left[  \left.  \phi\left(  y-\frac{\mu\left(  x_{0}\right)
\Delta+L\left(  \Delta\right)  }{\sigma(x_{0})\sqrt{\Delta}}\right)
\right\vert L\left(  \Delta\right)  \right]  . \label{int2_4}%
\end{align}
Plugging (\ref{int2_4}) into (\ref{int2_2}), the leading term $\Omega_{0}(y)$
can be finally calculated as
\[
\Omega_{0}(y)=E\left[  E\left[  \left.  \phi\left(  y-\frac{\mu\left(
x_{0}\right)  \Delta+L\left(  \Delta\right)  }{\sigma(x_{0})\sqrt{\Delta}%
}\right)  \right\vert L\left(  \Delta\right)  \right]  \right]  =\int%
_{0}^{+\infty}\phi\left(  y-\frac{\mu\left(  x_{0}\right)  \Delta+u}%
{\sigma(x_{0})\sqrt{\Delta}}\right)  p_{L\left(  \Delta\right)  }\left(
u\right)  du.
\]

\end{proof}

To calculate the higher-order terms $\Omega_{m}\left(  y\right)  $ for
$m\geq1$, we introduce the following notations. For $\ell\geq1$ and
$\mathbf{j}\left(  \ell\right)  =\left(  j_{1},j_{2},\ldots,j_{\ell}\right)  $
with $j_{i}\geq1$, we define
\begin{equation}
K_{\left(  \ell,\mathbf{j}\left(  \ell\right)  \right)  }\left(  z_{1}%
,z_{2}\right)  :=\left.  E\left(  \left.  \prod\limits_{i=1}^{\ell}X_{j_{i}%
+1}(\Delta)\right\vert W(\Delta),L\left(  \Delta\right)  \right)  \right\vert
_{W\left(  \Delta\right)  =z_{1}\sqrt{\Delta},L\left(  \Delta\right)  =z_{2}}.
\label{K(z)}%
\end{equation}
Meanwhile, for any bivariate differentiable function $u(x,y)$ defined on
$\mathbb{R}^{2}$, we introduce the following partial differential operators
with respect to the first variable:%
\begin{equation}
\mathcal{D}_{1}^{(1)}\left(  u(x,y)\right)  :=\frac{\partial u(x,y)}{\partial
x}-xu(x,y)\text{ and }\mathcal{D}_{1}^{(n)}\left(  u(x,y)\right)
:=\mathcal{D}_{1}^{(1)}\left(  \mathcal{D}_{1}^{(n-1)}\left(  u(x,y)\right)
\right)  \text{ for }n\geq2. \label{itera derivative}%
\end{equation}
The representations of $\Omega_{m}\left(  y\right)  $ for $m\geq1$ are given
in the following theorem.

\begin{theorem}
\label{theorem:omega_m} For any integer $m\geq1$, the high-order term
$\Omega_{m}\left(  y\right)  $\ in (\ref{pX_M_expan})\ admits the following
expression:%
\begin{equation}
\Omega_{m}(y)=\sum_{\left(  \ell,\left(  j_{1},j_{2},\cdots,j_{\ell}\right)
\right)  \in\mathcal{S}_{m}}\frac{(-1)^{^{\ell}}}{\ell!}\frac{1}{(\sigma
(x_{0})\sqrt{\Delta})^{\ell}}\int_{0}^{+\infty}\mathcal{D}_{1}^{(\ell)}\left(
K_{\left(  \ell,\mathbf{j}\left(  \ell\right)  \right)  }\left(  z_{1}%
,z_{2}\right)  \right)  \cdot\phi(z_{1})\cdot p_{L\left(  \Delta\right)
}(z_{2})dz_{2}, \label{Omega_m(y)}%
\end{equation}
where the index set $\mathcal{S}_{m}$ is defined in (\ref{Sm}), $p_{L\left(
\Delta\right)  }(\cdot)$ is the density function of the gamma process
$L(\Delta)$ given by (\ref{density_gamma process}) and
\begin{equation}
z_{1}=y-\frac{\mu\left(  x_{0}\right)  \Delta+z_{2}}{\sigma(x_{0})\sqrt
{\Delta}}. \label{z1}%
\end{equation}

\end{theorem}

\begin{proof}
We see from the definition of $\Phi_{m}(y)$ in (\ref{phim}) that
\begin{align}
\Omega_{m}(y)  &  =\mathbb{E}\Phi_{m}(y)=\sum_{\left(  \ell,\left(
j_{1},j_{2},\cdots,j_{\ell}\right)  \right)  \in\mathcal{S}_{m}}%
\mathbb{E}\left[  \frac{1}{\ell!}\frac{1}{(\sigma(x_{0})\sqrt{\Delta})^{\ell}%
}\frac{d^{\ell}\delta\left(  Y_{0}(\Delta)-y\right)  }{dx^{\ell}}%
\prod\limits_{i=1}^{\ell}X_{j_{i}+1}(\Delta)\right] \nonumber\\
&  =\sum_{\left(  \ell,\left(  j_{1},j_{2},\cdots,j_{\ell}\right)  \right)
\in\mathcal{S}_{m}}\frac{1}{\ell!}\frac{1}{(\sigma(x_{0})\sqrt{\Delta})^{\ell
}}\mathbb{E}\left[  \frac{d^{\ell}\delta\left(  Y_{0}(\Delta)-y\right)
}{dx^{\ell}}\prod\limits_{i=1}^{\ell}X_{j_{i}+1}(\Delta)\right]  .
\label{int2_5}%
\end{align}
For the generalized expectation in the last line of (\ref{int2_5}), according
to (\ref{generalised Esp}), we have
\begin{align*}
&  \mathbb{E}\left[  \frac{d^{\ell}\delta\left(  Y_{0}(\Delta)-y\right)
}{dx^{\ell}}\prod\limits_{i=1}^{\ell}X_{j_{i}+1}(\Delta)\right] \\
=  &  \frac{1}{2\pi}\int_{-\infty}^{+\infty}\int_{-\infty}^{+\infty}%
e^{-ivx}\frac{d^{\ell}\delta\left(  x-y\right)  }{dx^{\ell}}E\left[
\prod\limits_{i=1}^{\ell}X_{j_{i}+1}(\Delta)\cdot e^{ivY_{0}(\Delta)}\right]
dxdv\\
=  &  \int_{-\infty}^{+\infty}\left(  \frac{1}{2\pi}\int_{-\infty}^{+\infty
}e^{-ivx}\frac{d^{\ell}\delta\left(  x-y\right)  }{dx^{\ell}}dx\right)
E\left[  \prod\limits_{i=1}^{\ell}X_{j_{i}+1}(\Delta)\cdot e^{ivY_{0}(\Delta
)}\right]  dv.
\end{align*}
By the relations (\ref{integrate-by-part}) and (\ref{int2_1}), the above
equation is further calculated as
\begin{align}
&  \mathbb{E}\left[  \frac{d^{\ell}\delta\left(  Y_{0}(\Delta)-y\right)
}{dx^{\ell}}\prod\limits_{i=1}^{\ell}X_{j_{i}+1}(\Delta)\right] \nonumber\\
=  &  \int_{-\infty}^{+\infty}\left(  \frac{1}{2\pi}\left(  -1\right)  ^{\ell
}\int_{-\infty}^{+\infty}\delta\left(  x-y\right)  \frac{d^{\ell}e^{-ivx}%
}{dx^{\ell}}dx\right)  E\left[  \prod\limits_{i=1}^{\ell}X_{j_{i}+1}%
(\Delta)\cdot e^{ivY_{0}(\Delta)}\right]  dv\nonumber\\
=  &  \int_{-\infty}^{+\infty}\frac{1}{2\pi}\left(  iv\right)  ^{\ell}%
e^{-ivy}\cdot E\left[  \prod\limits_{i=1}^{\ell}X_{j_{i}+1}(\Delta)\cdot
e^{ivY_{0}(\Delta)}\right]  dv. \label{int2_6}%
\end{align}
Then following the definition of $Y_{0}\left(  \Delta\right)  $ in (\ref{Y_0})
and using the independence between Brownian motion $W(t)$ and gamma process
$L(t)$, the expectation in the last line of (\ref{int2_6}) is calculated as%
\begin{align}
&  E\left[  \prod\limits_{i=1}^{\ell}X_{j_{i}+1}(\Delta)\cdot e^{ivY_{0}%
(\Delta)}\right] \nonumber\\
=  &  E\left[  \prod\limits_{i=1}^{\ell}X_{j_{i}+1}(\Delta)\cdot e^{iv\left(
\frac{W(\Delta)}{\sqrt{\Delta}}+\frac{\mu\left(  x_{0}\right)  \Delta+L\left(
\Delta\right)  }{\sigma(x_{0})\sqrt{\Delta}}\right)  }\right] \nonumber\\
=  &  \int_{0}^{\infty}\int_{-\infty}^{\infty}E\left.  \left[  \left.
\prod\limits_{i=1}^{\ell}X_{j_{i}+1}(\Delta)\cdot e^{iv\left(  z_{1}+\frac
{\mu\left(  x_{0}\right)  \Delta+z_{2}}{\sigma(x_{0})\sqrt{\Delta}}\right)
}\right\vert W(\Delta),L\left(  \Delta\right)  \right]  \right\vert
_{W(\Delta)=z_{1}\sqrt{\Delta},L\left(  \Delta\right)  =z_{2}}\nonumber\\
&  \times\phi(z_{1})p_{L\left(  \Delta\right)  }(z_{2})dz_{1}dz_{2}\text{.}
\label{int2_7}%
\end{align}
By the definition of $K_{\left(  \ell,\mathbf{j}\left(  \ell\right)  \right)
}\left(  z_{1},z_{2}\right)  $\ in (\ref{K(z)}), the expectation
(\ref{int2_7}) can be expressed as%
\begin{equation}
E\left[  \prod\limits_{i=1}^{\ell}X_{j_{i}+1}(\Delta)\cdot e^{ivY_{0}(\Delta
)}\right]  =\int_{0}^{\infty}\int_{-\infty}^{\infty}e^{iv\left(  z_{1}%
+\frac{\mu\left(  x_{0}\right)  \Delta+z_{2}}{\sigma(x_{0})\sqrt{\Delta}%
}\right)  }K_{\left(  \ell,\mathbf{j}\left(  \ell\right)  \right)  }\left(
z_{1},z_{2}\right)  \phi(z_{1})p_{L\left(  \Delta\right)  }(z_{2})dz_{1}%
dz_{2}. \label{int2_8}%
\end{equation}
Plugging (\ref{int2_8}) into (\ref{int2_6}), we obtain that%
\begin{align}
&  \text{\ }\mathbb{E}\left[  \frac{d^{\ell}\delta\left(  Y_{0}(\Delta
)-y\right)  }{dx^{\ell}}\prod\limits_{i=1}^{\ell}X_{j_{i}+1}(\Delta)\right]
\nonumber\\
=  &  \int_{-\infty}^{+\infty}\frac{1}{2\pi}\left(  iv\right)  ^{\ell}%
e^{-ivy}\int_{0}^{\infty}\int_{-\infty}^{\infty}e^{iv\left(  z_{1}+\frac
{\mu\left(  x_{0}\right)  \Delta+z_{2}}{\sigma(x_{0})\sqrt{\Delta}}\right)
}K_{\left(  \ell,\mathbf{j}\left(  \ell\right)  \right)  }\left(  z_{1}%
,z_{2}\right)  \phi(z_{1})p_{L\left(  \Delta\right)  }(z_{2})dz_{1}%
dz_{2}dv\nonumber\\
=  &  \int_{0}^{+\infty}\int_{-\infty}^{\infty}\frac{1}{2\pi}e^{-ivy}%
p_{L\left(  \Delta\right)  }(z_{2})\left(  \int_{-\infty}^{\infty}%
\frac{\partial^{\ell}e^{iv\left(  z_{1}+\frac{\mu\left(  x_{0}\right)
\Delta+z_{2}}{\sigma(x_{0})\sqrt{\Delta}}\right)  }}{\partial z_{1}^{\ell}%
}K_{\left(  \ell,\mathbf{j}\left(  \ell\right)  \right)  }\left(  z_{1}%
,z_{2}\right)  \phi(z_{1})dz_{1}\right)  dvdz_{2}. \label{int2_9}%
\end{align}
Using integration by parts, the last line of (\ref{int2_9}) can be further
calculated as%
\begin{align}
&  \text{\ }\int_{0}^{+\infty}\int_{-\infty}^{\infty}\frac{1}{2\pi}%
e^{-ivy}p_{L\left(  \Delta\right)  }(z_{2})\int_{-\infty}^{\infty}%
(-1)^{^{\ell}}e^{iv\left(  z_{1}+\frac{\mu\left(  x_{0}\right)  \Delta+z_{2}%
}{\sigma(x_{0})\sqrt{\Delta}}\right)  }\frac{\partial^{^{\ell}}\left(
K_{\left(  \ell,\mathbf{j}\left(  \ell\right)  \right)  }\left(  z_{1}%
,z_{2}\right)  \phi(z_{1})\right)  }{\partial z_{1}^{^{\ell}}}dz_{1}%
dvdz_{2}\nonumber\\
=  &  \int_{0}^{+\infty}\int_{-\infty}^{\infty}\left(  \int_{-\infty}^{\infty
}\frac{1}{2\pi}e^{-ivy}e^{iv\left(  z_{1}+\frac{\mu\left(  x_{0}\right)
\Delta+z_{2}}{\sigma(x_{0})\sqrt{\Delta}}\right)  }dv\right) \nonumber\\
&  \times(-1)^{^{\ell}}\frac{\partial^{^{\ell}}\left(  K_{\left(
\ell,\mathbf{j}\left(  \ell\right)  \right)  }\left(  z_{1},z_{2}\right)
\phi(z_{1})\right)  }{\partial z_{1}^{^{\ell}}}p_{L\left(  \Delta\right)
}(z_{2})dz_{1}dz_{2}. \label{int2_10}%
\end{align}
By the relation (\ref{delta_exp}), it follows that
\[
\frac{1}{2\pi}\int_{-\infty}^{+\infty}e^{-ivy}e^{iv\left(  z_{1}+\frac
{\mu\left(  x_{0}\right)  \Delta+z_{2}}{\sigma(x_{0})\sqrt{\Delta}}\right)
}dv=\frac{1}{2\pi}\int_{-\infty}^{+\infty}e^{iv\left(  z_{1}+\frac{\mu\left(
x_{0}\right)  \Delta+z_{2}}{\sigma(x_{0})\sqrt{\Delta}}-y\right)  }%
dv=\delta\left(  z_{1}+\frac{\mu\left(  x_{0}\right)  \Delta+z_{2}}%
{\sigma(x_{0})\sqrt{\Delta}}-y\right)  ,
\]
Plugging the above equation into (\ref{int2_10}), we obtain that%
\begin{align}
&  \text{\ }\mathbb{E}\left[  \frac{d^{\ell}\delta\left(  Y_{0}(\Delta
)-y\right)  }{dx^{\ell}}\prod\limits_{i=1}^{\ell}X_{j_{i}+1}(\Delta)\right]
\nonumber\\
=  &  \int_{0}^{+\infty}\int_{-\infty}^{\infty}\delta\left(  z_{1}+\frac
{\mu\left(  x_{0}\right)  \Delta+z_{2}}{\sigma(x_{0})\sqrt{\Delta}}-y\right)
(-1)^{^{\ell}}\frac{\partial^{^{\ell}}\left(  K_{\left(  \ell,\mathbf{j}%
\left(  \ell\right)  \right)  }\left(  z_{1},z_{2}\right)  \phi(z_{1})\right)
}{\partial z_{1}^{^{\ell}}}p_{L\left(  \Delta\right)  }(z_{2})dz_{1}%
dz_{2}\nonumber\\
=  &  \int_{0}^{+\infty}(-1)^{^{\ell}}\left.  \frac{\partial^{^{\ell}}\left(
K_{\left(  \ell,\mathbf{j}\left(  \ell\right)  \right)  }\left(  z_{1}%
,z_{2}\right)  \phi(z_{1})\right)  }{\partial z_{1}^{^{\ell}}}\right\vert
_{z_{1}=y-\frac{\mu\left(  x_{0}\right)  \Delta+z_{2}}{\sigma(x_{0}%
)\sqrt{\Delta}}}\cdot p_{L\left(  \Delta\right)  }(z_{2})dz_{2}.
\label{int2_11}%
\end{align}
Thus, plugging (\ref{int2_11}) into (\ref{int2_5}), we obtain that%
\begin{align}
\Omega_{m}(y)  &  =\sum_{\left(  \ell,\left(  j_{1},j_{2},\cdots,j_{\ell
}\right)  \right)  \in\mathcal{S}_{m}}\frac{(-1)^{^{\ell}}}{\ell!}\frac
{1}{(\sigma(x_{0})\sqrt{\Delta})^{\ell}}\nonumber\\
&  \text{ \ \ \ \ \ \ \ \ \ \ \ \ \ }\times\int_{0}^{+\infty}\left.
\frac{\partial^{^{\ell}}\left(  K_{\left(  \ell,\mathbf{j}\left(  \ell\right)
\right)  }\left(  z_{1},z_{2}\right)  \phi(z_{1})\right)  }{\partial
z_{1}^{^{\ell}}}\right\vert _{z_{1}=y-\frac{\mu\left(  x_{0}\right)
\Delta+z_{2}}{\sigma(x_{0})\sqrt{\Delta}}}p_{L\left(  \Delta\right)  }%
(z_{2})dz_{2}. \label{int2_12}%
\end{align}

From the definition (\ref{itera derivative}), for any bivariate differentiable
function $u\left(  z_{1},z_{2}\right)  $, we have%
\[
\frac{\partial}{\partial z_{1}}\left(  u(z_{1},z_{2})\phi(z_{1})\right)
=\left(  \frac{\partial u(z_{1},z_{2})}{\partial z_{1}}-z_{1}u(z_{1}%
,z_{2})\right)  \phi(z_{1})\equiv\mathcal{D}_{1}^{(1)}\left(  u(z_{1}%
,z_{2})\right)  \cdot\phi(z_{1}),
\]
and%
\[
\frac{\partial^{\ell}}{\partial z_{1}^{\ell}}\left(  u(z_{1},z_{2})\phi
(z_{1})\right)  =\mathcal{D}_{1}^{(\ell)}\left(  u(z_{1},z_{2})\right)
\cdot\phi(z_{1})
\]
for any $\ell\geq1$, from which we obtain that%
\begin{equation}
\frac{\partial^{\ell}\left(  K_{\left(  \ell,\mathbf{j}\left(  \ell\right)
\right)  }\left(  z_{1},z_{2}\right)  \phi(z_{1})\right)  }{\partial
z_{1}^{\ell}}=\mathcal{D}_{1}^{(\ell)}\left(  K_{\left(  \ell,\mathbf{j}%
\left(  \ell\right)  \right)  }\left(  z_{1},z_{2}\right)  \right)  \cdot
\phi(z_{1}). \label{int2_13}%
\end{equation}
Plugging (\ref{int2_13}) into (\ref{int2_12}), we obtain the formula
(\ref{Omega_m(y)}).
\end{proof}

According to (\ref{Omega_m(y)}) in Theorem \ref{theorem:omega_m}, to calculate
the high-order terms $\Omega_{m}(y)$ for $m\geq1$, it suffices to derive the
bivariate function $K_{(\ell,\mathbf{j(}\ell))}(z_{1},z_{2})$ in (\ref{K(z)}),
which will be shown in Section \ref{sec:calculate_K}.

\section{Explicit calculation of $K_{\left(  \ell,\mathbf{j}\left(
\ell\right)  \right)  }\left(  z_{1},z_{2}\right)  $}

\label{sec:calculate_K}

In this section, we explicitly derive the function $K_{\left(  \ell
,\mathbf{j}\left(  \ell\right)  \right)  }\left(  z_{1},z_{2}\right)  $ for
every fixed $\ell\geq1$ and $\mathbf{j}\left(  \ell\right)  =\left(
j_{1},j_{2},\ldots,j_{\ell}\right)  $ with $j_{i}\geq1$ in (\ref{K(z)}), from
which we can evaluate $\Omega_{m}\left(  y\right)  $ for $m\geq1$\ in
(\ref{Omega_m(y)}) and obtain the approximated transition density $p_{X\left(
\Delta\right)  }^{\left(  M\right)  }\left(  \left.  x\right\vert
x_{0};\bm{\theta}\right)  $ in (\ref{pX_M_expan}). Moreover, for illustration
purpose, the pure jump OU model, constant diffusion model and square-root
diffusion model are introduced as examples of SDE (\ref{model}) to exhibit the
first several expansion terms of $\left\{  \Omega_{m}(y),m\geq0\right\}  $ in
(\ref{pX_M_expan}).

To present our algorithm for calculating the function $K_{\left(
\ell,\mathbf{j}\left(  \ell\right)  \right)  }\left(  z_{1},z_{2}\right)  $,
we introduce the following notation. For any integer $h\geq1$ and\ arbitrary
$h$-dimensional index $\mathbf{n}(h)=(n_{1},n_{2},\ldots,n_{h})$ with
nonnegative integers $n_{1},n_{2},\ldots,n_{h}$, we define the $h$-dimensional
vector%
\begin{equation}
\mathbf{L}^{\mathbf{n}(h)}(t):=\left(  L^{n_{1}}(t),L^{n_{2}}(t),\ldots
,L^{n_{h}}(t)\right)  \label{multi_gamma process}%
\end{equation}
by using the gamma process $L\left(  \cdot\right)  $. For example,
$\mathbf{L}^{\mathbf{n}(1)}(t)=(L(t))$ when $h=1$ and $\mathbf{n}(1)=(1)$, and
$\mathbf{L}^{\mathbf{n}(2)}(t)=(1,L(t))$ when $h=2$ and $\mathbf{n}\left(
2\right)  =(0,1)$.

For any $h$-dimensional index $\mathbf{i}\left(  h\right)  =\left(
i_{1},i_{2},\ldots,i_{h}\right)  $ with $i_{1},i_{2},\ldots,i_{h}\in\left\{
0,1\right\}  $ and $h$-dimensional vector $\mathbf{L}^{\mathbf{n}(h)}(t)$ in
(\ref{multi_gamma process}), we define an iterated stochastic integral as
\begin{align}
&  \text{ \ \ \ \ \ }\mathbf{I}_{\mathbf{i}\left(  h\right)  ,\mathbf{L}%
^{\mathbf{n}(h)}}(\Delta)\nonumber\\
&  :=\int_{0}^{\Delta}\int_{0}^{s_{h}}\cdots\int_{0}^{s_{2}}L^{n_{1}}%
(s_{1})\cdots L^{n_{h-1}}(s_{h-1})L^{n_{h}}(s_{h})dW_{i_{1}}(s_{1})\cdots
dW_{i_{h-1}}(s_{h-1})dW_{i_{h}}(s_{h}), \label{multi_integral}%
\end{align}
where $W_{0}(t):=t$ and $W_{1}(t):=W(t)$. For example, we have
\begin{align*}
&  \left.  \mathbf{I}_{(0),\mathbf{L}^{(0)}}(\Delta)=\Delta\text{, }%
\mathbf{I}_{(1),\mathbf{L}^{(0)}}(\Delta)=W(\Delta)\text{, }\mathbf{I}%
_{(0),\mathbf{L}^{(1)}}(\Delta)=\int_{0}^{\Delta}L(s_{1})ds_{1}\text{,
}\mathbf{I}_{(1),\mathbf{L}^{(1)}}(\Delta)=\int_{0}^{\Delta}L\left(
s_{1}\right)  dW(s_{1})\text{,}\right. \\
&  \left.  \mathbf{I}_{(0,0),\mathbf{L}^{(0,0)}}(\Delta)=\int_{0}^{\Delta}%
\int_{0}^{s_{2}}ds_{1}ds_{2}\text{, }\mathbf{I}_{(0,1),\mathbf{L}^{(0,0)}%
}(\Delta)=\int_{0}^{\Delta}\int_{0}^{s_{2}}ds_{1}dW(s_{2})\text{,}\right.
\end{align*}
and%
\[
\mathbf{I}_{(1,0),\mathbf{L}^{(0,1)}}(\Delta)=\int_{0}^{\Delta}\int_{0}%
^{s_{2}}L(s_{2})dW(s_{1})ds_{2}\text{, }\mathbf{I}_{(1,1),\mathbf{L}^{(1,0)}%
}(\Delta)=\int_{0}^{\Delta}\int_{0}^{s_{2}}L(s_{1})dW(s_{1})dW(s_{2}).
\]
We notice that the iterated stochastic integral $\mathbf{I}_{\mathbf{i}\left(
h\right)  ,\mathbf{L}^{\mathbf{n}(h)}}(\Delta)$ defined by
(\ref{multi_integral}) involves two independent processes, i.e., the Brownian
motion $W(\cdot)$ and the gamma process $L(\cdot)$. Such independence will
simplify the calculation related to $\mathbf{I}_{\mathbf{i}\left(  h\right)
,\mathbf{L}^{\mathbf{n}(h)}}(\Delta)$ as seen below.

In order to clarify the procedures of calculating $K_{\left(  \ell
,\mathbf{j}\left(  \ell\right)  \right)  }\left(  z_{1},z_{2}\right)  $ in
(\ref{K(z)}), we briefly outline a general algorithm before the detailed
descriptions below, which can be implemented by traditional symbolic
softwares, e.g., Wolfram Mathematica.

\begin{algorithm}[H]
\label{int_algorithm}
\caption{ Framework of calculating $K_{\left(  \ell,\mathbf{j}\left(  \ell\right)
\right)  }\left(  z_1,z_2\right)  $ in (\ref{K(z)}).}
\label{alg:Framwork}
\begin{enumerate}[Step 1]
\item Convert the multiplication of the expansion terms in $K_{\left(
\ell,\mathbf{j}\left(  \ell\right)  \right)  }\left(  z_1,z_2\right)  $, i.e.,
$\prod_{i=1}^{\ell}X_{j_{i}+1}(\Delta)$, to a linear combination of
iterated It\^{o} integrals as defined in
(\ref{multi_integral});
\item Simplify the conditional expectation of the iterated It\^{o} integral
$\mathbf{I}_{\mathbf{i}\left(  h\right)  ,\mathbf{L}^{\mathbf{n}(h)}}(\Delta)$ via Brownian bridge;
\item Compute the conditional expectation of the result from Step 2 with respect to
the gamma process.
\end{enumerate}
\end{algorithm}

In the following Sections \ref{subsec:Step1}, \ref{subsec:Step2} and
\ref{subsec:Step3}, we give the detailed descriptions of Steps 1, 2 and 3 in
the above algorithm respectively. In Section \ref{subsec:examples}, we
consider three examples of SDE (\ref{model}) for illustrations.

\subsection{Conversion of the multiplication $\prod\nolimits_{i=1}^{\ell
}X_{j_{i}+1}(\Delta)$ into a linear combination of iterated It\^{o} integrals}

\label{subsec:Step1}

First, we illustrate that the multiplication of iterated It\^{o} integrals
defined in (\ref{multi_integral}) can be converted into a linear combination
of the iterated It\^{o} integrals taking the same form as in
(\ref{multi_integral}).

Given an index $\mathbf{i}\left(  h\right)  =\left(  i_{1},i_{2},\ldots
,i_{h}\right)  $, we denote by $\mathbf{i}\left(  h\right)  -$ the index
obtained from deleting the last element of index $\mathbf{i}\left(  h\right)
$, i.e.,
\[
\mathbf{i}\left(  h\right)  -:=\left(  i_{1},i_{2},\ldots,i_{h-1}\right)  .
\]
Similarly, we denote by
\[
\mathbf{L}^{\mathbf{n}(h)-}(t):=\left(  L^{n_{1}}(t),L^{n_{2}}(t),\ldots
,L^{n_{h-1}}(t)\right)
\]
the $\left(  h-1\right)  $-dimensional vector obtained from deleting the last
element of $\mathbf{L}^{\mathbf{n}(h)}(t)$ in (\ref{multi_gamma process}).
Consequently, the iterated It\^{o} Integral $\mathbf{I}_{\mathbf{i}\left(
h\right)  \mathbf{-},\mathbf{L}^{\mathbf{n}(h)-}}(\Delta)$ can be defined as
\begin{align*}
\mathbf{I}_{\mathbf{i}\left(  h\right)  \mathbf{-},\mathbf{L}^{\mathbf{n}%
(h)-}}(\Delta)  &  :=\int_{0}^{\Delta}\int_{0}^{s_{h-1}}\cdots\int_{0}^{s_{2}%
}L^{n_{1}}(s_{1})\cdots L^{n_{h-2}}(s_{h-2})\\
&  \quad\quad\times L^{n_{h-1}}(s_{h-1})dW_{i_{1}}(s_{1})\cdots dW_{i_{h-2}%
}(s_{h-2})dW_{i_{h-1}}(s_{h-1}).
\end{align*}

For two fixed positive integers $h,q$ and the gamma process $L\left(
\cdot\right)  $, we consider the $h$-dimensional vector $\mathbf{L}%
^{\mathbf{n}(h)}(t)$ and $q$-dimensional vector $\mathbf{L}^{\mathbf{m}%
(q)}(t)$,%
\[
\mathbf{L}^{\mathbf{n}(h)}(t)=\left(  L^{n_{1}}(t),L^{n_{2}}(t),\ldots
,L^{n_{h}}(t)\right)  \text{ and }\mathbf{L}^{\mathbf{m}(q)}(t)=\left(
L^{m_{1}}(t),L^{m_{2}}(t),\ldots,L^{m_{q}}(t)\right)
\]
for some indices $\mathbf{n}\left(  h\right)  \mathbf{=}\left(  n_{1}%
,n_{2},\ldots,n_{h}\right)  $ and $\mathbf{m}\left(  q\right)  \mathbf{=}%
\left(  m_{1},m_{2},\ldots,m_{q}\right)  $ with nonnegative integers $n_{1}$,
$n_{2},\ldots,n_{h},m_{1},m_{2},\ldots,m_{q}$. Then for two indices
$\mathbf{i}\left(  h\right)  \mathbf{=}\left(  i_{1},i_{2},\ldots
,i_{h}\right)  $ and $\mathbf{j}\left(  q\right)  \mathbf{=}\left(
j_{1},j_{2},\ldots,j_{q}\right)  $ with $i_{1},i_{2,}\ldots,i_{h}$, $j_{1}$,
$j_{2},\ldots,j_{q}\in\left\{  0,1\right\}  $, the product of two iterated
It\^{o} integrals $\mathbf{I}_{\mathbf{i}\left(  h\right)  ,\mathbf{L}%
^{\mathbf{n}(h)}}(\Delta)$ and $\mathbf{I}_{\mathbf{j}\left(  q\right)
,\mathbf{L}^{\mathbf{m}(q)}}(\Delta)$ satisfies the following iterative
relation%
\begin{align}
&  \mathbf{I}_{\mathbf{i}\left(  h\right)  ,\mathbf{L}^{\mathbf{n}(h)}}%
(\Delta)\mathbf{I}_{\mathbf{j}\left(  q\right)  ,\mathbf{L}^{\mathbf{m}(q)}%
}(\Delta)\nonumber\\
=  &  \left[  \int_{0}^{\Delta}\mathbf{I}_{\mathbf{i}\left(  h\right)
-,\mathbf{L}^{\mathbf{n}(h)-}}(s_{1})\cdot L^{n_{h}}(s_{1})dW_{i_{h}}%
(s_{1})\right]  \cdot\left[  \int_{0}^{\Delta}\mathbf{I}_{\mathbf{j}\left(
q\right)  -,\mathbf{L}^{\mathbf{m}(q)-}}(s_{1})\cdot L^{m_{q}}(s_{1}%
)dW_{j_{q}}(s_{1})\right] \nonumber\\
=  &  \int_{0}^{\Delta}\mathbf{I}_{\mathbf{i}\left(  h\right)  ,\mathbf{L}%
^{\mathbf{n}(h)}}(s_{1})\mathbf{I}_{\mathbf{j}\left(  q\right)  -,\mathbf{L}%
^{\mathbf{m}(q)-}}(s_{1})\cdot L^{m_{q}}(s_{1})dW_{j_{q}}(s_{1})\nonumber\\
&  +\int_{0}^{\Delta}\mathbf{I}_{\mathbf{i}\left(  h\right)  -,\mathbf{L}%
^{\mathbf{n}(h)-}}(s_{1})\mathbf{I}_{\mathbf{j}\left(  q\right)
,\mathbf{L}^{\mathbf{m}(q)}}(s_{1})\cdot L^{n_{h}}(s_{1})dW_{i_{h}}%
(s_{1})\nonumber\\
&  +\int_{0}^{\Delta}\mathbf{I}_{\mathbf{i}\left(  h\right)  -,\mathbf{L}%
^{\mathbf{n}(h)-}}(s_{1})\mathbf{I}_{\mathbf{j}\left(  q\right)
-,\mathbf{L}^{\mathbf{m}(q)-}}(s_{1})\cdot L^{n_{h}+m_{q}}(s_{1}%
)\cdot1_{\left\{  i_{h}=j_{q}=1\right\}  }ds_{1},
\label{product_multi_integral}%
\end{align}
\textbf{ }where the second equation follows from the It\^{o} product formula%
\begin{align*}
&  \int_{0}^{\Delta}f(s_{1})dW_{i_{1}}\left(  s_{1}\right)  \cdot\int%
_{0}^{\Delta}g(s_{1})dW_{j_{1}}\left(  s_{1}\right) \\
=  &  \int_{0}^{\Delta}\int_{0}^{s_{1}}f(s_{2})dW_{i_{1}}\left(  s_{2}\right)
g(s_{1})dW_{j_{1}}\left(  s_{1}\right)  +\int_{0}^{\Delta}\int_{0}^{s_{1}%
}g(s_{2})dW_{j_{1}}\left(  s_{2}\right)  f(s_{1})dW_{i_{1}}\left(
s_{1}\right) \\
&  +\int_{0}^{\Delta}f(s_{1})g(s_{1})1_{\left\{  i_{1}=j_{1}=1\right\}
}ds_{1}%
\end{align*}
and $1_{\left\{  i_{h}=j_{q}=1\right\}  }$ is the indicator function defined
as
\[
1_{\{i_{h}=j_{q}=1\}}=\left\{
\begin{array}
[c]{l}%
1,\text{ if }i_{h}=j_{q}=1,\\
0,\text{ otherwise.}%
\end{array}
\right.
\]
By iterative applications of the relation (\ref{product_multi_integral}), the
product of $I_{\mathbf{i}\left(  h\right)  ,\mathbf{L}^{\mathbf{n}(h)}}%
(\Delta)$ and $I_{\mathbf{j}\left(  q\right)  ,\mathbf{L}^{\mathbf{m}(q)}%
}(\Delta)$ can be expressed as a linear combination of the iterated It\^{o}
integrals defined by (\ref{multi_integral}).

Next, we show that the expansion terms $X_{j_{1}+1}(\Delta),X_{j_{2}+1}%
(\Delta),\ldots,X_{j_{\ell}+1}(\Delta)$ in (\ref{K(z)}) can be expressed as a
linear combination of iterated It\^{o} integrals $I_{\mathbf{i}\left(
h\right)  ,\mathbf{L}^{\mathbf{n}(h)}}(\Delta)$ defined in
(\ref{multi_integral}), with coefficients depending on $\mu\left(
x_{0}\right)  $, $\sigma(x_{0})$ and their\ higher-order derivatives evaluated
at $x_{0}$. Based on this, it follows from (\ref{product_multi_integral}) that
the multiplication $\prod\nolimits_{i=1}^{\ell}X_{j_{i}+1}(\Delta)$ can be
converted into a linear combination of $I_{\mathbf{i}\left(  h\right)
,\mathbf{L}^{\mathbf{n}(h)}}(\Delta)$ defined in (\ref{multi_integral}). To do
this, in what follows, we illustrate that $X_{m}(\Delta)$ admits the
aforementioned linear combination form\ for $m\geq1$. By the notation
(\ref{multi_integral}), $X_{1}(\Delta)$ in (\ref{X_1}) can be written as%
\begin{equation}
X_{1}(\Delta)=\mu\left(  x_{0}\right)  \mathbf{I}_{(0),\mathbf{L}^{(0)}%
}(\Delta)+\sigma\left(  x_{0}\right)  \mathbf{I}_{(1),\mathbf{L}^{(0)}}%
(\Delta)+L(\Delta), \label{X1_eg}%
\end{equation}
which admits the linear combination form. For $m\geq1$, we notice from
(\ref{X_m}) that
\begin{equation}
X_{m+1}(\Delta)=\int_{0}^{\Delta}\mu_{m}(s)ds+\int_{0}^{\Delta}\sigma
_{m}(s)dW(s), \label{Xm_eg}%
\end{equation}
with $\mu_{m}(s)$ and $\sigma_{m}(s)$ defined by (\ref{mu_m(t)}) --
(\ref{sigmam(t)}). Since both $\mu_{m}(s)$ and $\sigma_{m}(s)$ are linear
combinations of the products of the terms chosen among $\left\{
X_{1}(s),X_{2}(s),\ldots,X_{m}(s)\right\}  $ (cf. Section
\ref{subsec: expan_tran_density}), by iterative applications of
(\ref{product_multi_integral}), (\ref{X1_eg}) and (\ref{Xm_eg}), we can also
derive $X_{m+1}(\Delta)$ for $m\geq1$ as a linear combination of iterated
It\^{o} integrals $\mathbf{I}_{\mathbf{i}\left(  h\right)  ,\mathbf{L}%
^{\mathbf{n}(h)}}(\Delta)$ for $h\leq m+1$\ formed as (\ref{multi_integral}),
with the coefficients depending on $\mu\left(  x_{0}\right)  $, $\sigma
(x_{0})$ and their\ higher-order derivatives evaluated at $x_{0}$.

In summary, to calculate $K_{\left(  \ell,\mathbf{j}\left(  \ell\right)
\right)  }\left(  z_{1},z_{2}\right)  $ in (\ref{K(z)}) for every fixed
$\ell\geq1$ and $\mathbf{j}\left(  \ell\right)  =\left(  j_{1},j_{2}%
,\ldots,j_{\ell}\right)  $ with $j_{i}\geq1$, it suffices to focus on the
following type of conditional expectation
\begin{equation}
E\left.  \left(  \mathbf{I}_{\mathbf{i}\left(  h\right)  ,\mathbf{L}%
^{\mathbf{n}(h)}}(\Delta)\left\vert W\left(  \Delta\right)  ,L\left(
\Delta\right)  \right.  \right)  \right\vert _{W\left(  \Delta\right)
=z_{1}\sqrt{\Delta},L\left(  \Delta\right)  =z_{2}} \label{bi_Cond_Eps}%
\end{equation}
with $\mathbf{I}_{\mathbf{i}\left(  h\right)  ,\mathbf{L}^{\mathbf{n}(h)}%
}(\Delta)$ defined by (\ref{multi_integral}).

\subsection{Simplification of the conditional expectation (\ref{bi_Cond_Eps})
via Brownian bridge}

\label{subsec:Step2}

Starting from this part, we focus on calculating the following conditional
expectation%
\begin{align}
&  E\left.  \left(  \mathbf{I}_{\mathbf{i}\left(  h\right)  ,\mathbf{L}%
^{\mathbf{n}(h)}}(\Delta)\left\vert W\left(  \Delta\right)  ,L\left(
\Delta\right)  \right.  \right)  \right\vert _{W\left(  \Delta\right)
=z_{1}\sqrt{\Delta},L\left(  \Delta\right)  =z_{2}}\nonumber\\
=  &  E\left(  \int_{0}^{\Delta}\int_{0}^{s_{h}}\cdots\int_{0}^{s_{2}}%
L^{n_{1}}(s_{1})\cdots L^{n_{h-1}}(s_{h-1})L^{n_{h}}(s_{h})\right. \nonumber\\
&  \left.  \left.  \quad\times dW_{i_{1}}(s_{1})\cdots dW_{i_{h-1}}%
(s_{h-1})dW_{i_{h}}(s_{h})\left\vert W\left(  \Delta\right)  ,L\left(
\Delta\right)  \right.  \right)  \right\vert _{W\left(  \Delta\right)
=z_{1}\sqrt{\Delta},L\left(  \Delta\right)  =z_{2}}, \label{bi_multi_Cond_Eps}%
\end{align}
where the iterated It\^{o} integral $\mathbf{I}_{\mathbf{i}\left(  h\right)
,\mathbf{L}^{\mathbf{n}(h)}}(\Delta)$ is defined by (\ref{multi_integral})
with $i_{1},i_{2},\ldots,i_{h}\in\left\{  0,1\right\}  $, $W_{0}(t)=t$ and
$W_{1}(t)=W(t)$.

To simplify (\ref{bi_multi_Cond_Eps}), we utilize the following representation
of Brownian bridge, i.e.,%
\begin{equation}
\left(  W(s)\left\vert W(\Delta)=z_{1}\sqrt{\Delta}\right.  \right)
\overset{d}{=}B^{z_{1}}(s):=B(s)-\frac{s}{\Delta}B\left(  \Delta\right)
+\frac{s}{\sqrt{\Delta}}z_{1} \label{BM bridge}%
\end{equation}
for $0\leq s\leq\Delta$, where the symbol \textquotedblleft$\overset{d}{=}%
$\textquotedblright\ means distributional identity and $B\left(  \cdot\right)
$ is a 1-dimensional standard Brownian motion. Then by the independence
between $W\left(  \cdot\right)  $ and $L(\cdot)$, (\ref{bi_multi_Cond_Eps})
can be equivalently expressed as%
\begin{align}
&  E\left.  \left(  \mathbf{I}_{\mathbf{i}\left(  h\right)  ,\mathbf{L}%
^{\mathbf{n}(h)}}(\Delta)\left\vert W\left(  \Delta\right)  ,L\left(
\Delta\right)  \right.  \right)  \right\vert _{W\left(  \Delta\right)
=z_{1}\sqrt{\Delta},L\left(  \Delta\right)  =z_{2}}\nonumber\\
=  &  E\left(  \int_{0}^{\Delta}\int_{0}^{s_{h}}\cdots\int_{0}^{s_{2}}%
L^{n_{1}}(s_{1})\cdots L^{n_{h-1}}(s_{h-1})\right. \nonumber\\
&  \left.  \left.  \left.  \times L^{n_{h}}(s_{h})dB_{i_{1}}^{z_{1}}%
(s_{1})\cdots dB_{i_{h-1}}^{z_{1}}(s_{h-1})dB_{i_{h}}^{z_{1}}(s_{h}%
)\right\vert L\left(  \Delta\right)  \right)  \right\vert _{L\left(
\Delta\right)  =z_{2}}, \label{int3_1}%
\end{align}
where $B{}_{1}^{z_{1}}(s):=B^{z_{1}}(s)$ and $B{}_{0}^{z_{1}}(s):=s$.
Therefore, we only need focus on the conditional expectation%
\begin{equation}
E\left(  \left.  \int_{0}^{\Delta}\cdots\int_{0}^{s_{2}}L^{n_{1}}(s_{1})\cdots
L^{n_{h}}(s_{h})dB_{i_{1}}^{z_{1}}(s_{1})\cdots dB_{i_{h}}^{z_{1}}%
(s_{h})\right\vert L\left(  \Delta\right)  \right)  , \label{int3_2}%
\end{equation}
from which (\ref{int3_1}) can be obtained by letting $L\left(  \Delta\right)
=z_{2}$. For the sake of simplicity, we denote by $E_{L}\left(  \cdot\right)
:=E\left(  \left.  \cdot\right\vert L\left(  \Delta\right)  \right)  $ the
conditional expectation given $L\left(  \Delta\right)  $ hereafter. By
plugging (\ref{BM bridge}) into (\ref{int3_2}), we obtain that%
\begin{align}
&  \text{ \ \ \ \ }E\left(  \left.  \int_{0}^{\Delta}\int_{0}^{s_{h}}%
\cdots\int_{0}^{s_{2}}L^{n_{1}}(s_{1})\cdots L^{n_{h-1}}(s_{h-1})L^{n_{h}%
}(s_{h})dB_{i_{1}}^{z_{1}}(s_{1})\cdots dB_{i_{h-1}}^{z_{1}}(s_{h-1}%
)dB_{i_{h}}^{z_{1}}(s_{h})\right\vert L\left(  \Delta\right)  \right)
\nonumber\\
&  =E_{L}\left(  \int_{0}^{\Delta}\int_{0}^{s_{h}}\cdots\int_{0}^{s_{2}%
}L^{n_{1}}(s_{1})dB_{i_{1}}^{z}(s_{1})\cdots L^{n_{h-1}}(s_{h-1})dB_{i_{h-1}%
}^{z_{1}}(s_{h-1})L^{n_{h}}(s_{h})dB_{i_{h}}^{z_{1}}(s_{h})\right) \nonumber\\
&  =E_{L}\left(  \int_{0}^{\Delta}\int_{0}^{s_{h}}\cdots\int_{0}^{s_{2}%
}L^{n_{1}}(s_{1})\left\{  1_{\left\{  i_{1}=1\right\}  }\left(  dB(s_{1}%
)-\frac{B\left(  \Delta\right)  }{\Delta}ds_{1}+\frac{z_{1}}{\sqrt{\Delta}%
}ds_{1}\right)  +1_{\left\{  i_{1}=0\right\}  }ds_{1}\right\}  \right.
\nonumber\\
&  \text{ \ \ \ \ \ \ }\left.  \text{ }\times\cdots\times L^{n_{h-1}}%
(s_{h-1})\left\{  1_{\left\{  i_{h-1}=1\right\}  }\left(  dB(s_{h-1}%
)-\frac{B\left(  \Delta\right)  }{\Delta}ds_{h-1}+\frac{z_{1}}{\sqrt{\Delta}%
}ds_{h-1}\right)  +1_{\left\{  i_{h-1}=0\right\}  }ds_{h-1}\right\}  \right.
\nonumber\\
&  \text{ \ \ \ \ \ \ }\left.  \text{ }\times L^{n_{h}}(s_{h})\left\{
1_{\left\{  i_{h}=1\right\}  }\left(  dB(s_{h})-\frac{B\left(  \Delta\right)
}{\Delta}ds_{h}+\frac{z_{1}}{\sqrt{\Delta}}ds_{h}\right)  +1_{\left\{
i_{h}=0\right\}  }ds_{h}\right\}  \right)  . \label{int3_3}%
\end{align}

In order to derive the explicit expression of (\ref{int3_3}), for any
$h$-dimensional index $\mathbf{i}\left(  h\right)  =\left(  i_{1},i_{2}%
,\ldots,i_{h}\right)  $ with $i_{1},i_{2},\ldots,i_{h}\in\left\{  0,1\right\}
$ and $h$-dimensional vector $\mathbf{L}^{\mathbf{n}(h)}(t)$ in
(\ref{multi_gamma process}), we define an iterated stochastic integral%
\begin{align}
&  \text{ \ \ \ \ \ }\mathbf{J}_{\mathbf{i}\left(  h\right)  ,\mathbf{L}%
^{\mathbf{n}(h)}}(\Delta)\nonumber\\
&  :=\int_{0}^{\Delta}\int_{0}^{s_{h}}\cdots\int_{0}^{s_{2}}L^{n_{1}}%
(s_{1})\cdots L^{n_{h-1}}(s_{h-1})L^{n_{h}}(s_{h})dB_{i_{1}}(s_{1})\cdots
B_{i_{h-1}}(s_{h-1})dB_{i_{h}}(s_{h}), \label{J_multi_integral}%
\end{align}
where $B_{0}(t):=t$ and $B_{1}(t):=B(t)$, with the Brownian motion $B(t)$
introduced in (\ref{BM bridge}). Then we fully expand the product of the
differential forms in the last equation of (\ref{int3_3}) and find that it
suffices to calculate the following two kinds of conditional expectations%
\begin{equation}
\left(  \frac{z_{1}}{\sqrt{\Delta}}\right)  ^{k_{1}}E_{L}\left(
\mathbf{J}_{\mathbf{i}\left(  h\right)  ,\mathbf{L}^{\mathbf{n}(h)}}%
(\Delta)\right)  \label{int3_5}%
\end{equation}
and%
\begin{equation}
\left(  \frac{z_{1}}{\sqrt{\Delta}}\right)  ^{k_{2}}E_{L}\left(
B(\Delta)^{k_{3}}\cdot\mathbf{J}_{\mathbf{i}\left(  h\right)  ,\mathbf{L}%
^{\mathbf{n}(h)}}(\Delta)\right)  , \label{int3_6}%
\end{equation}
where the integers $k_{1},k_{2},k_{3}\in\left\{  0,1,\ldots,h\right\}  $
satisfying the condition $k_{2}+k_{3}\leq h$. To precede, we notice the
following relation%
\begin{align}
B\left(  \Delta\right)  \mathbf{J}_{\mathbf{i}\left(  h\right)  ,\mathbf{L}%
^{\mathbf{n}(h)}}(\Delta)  &  =\left.  \sum_{m=1}^{h+1}\mathbf{J}%
_{\mathbf{(}i_{1},\ldots,i_{m-1},1,i_{m},\ldots,i_{h}\mathbf{)},(L^{n_{1}%
}(\Delta),\ldots,L^{n_{m-1}}(\Delta),1,L^{n_{m}}(\Delta),\ldots,L^{n_{h}%
}(\Delta))}(\Delta)\right. \nonumber\\
\text{ \ \ \ }  &  \left.  \text{ \ \ }+\sum_{m=1}^{h}1_{\left\{
i_{m}=1\right\}  }\mathbf{J}_{\mathbf{(}i_{1},\ldots,i_{m-1},0,i_{m+1}%
,\ldots,i_{h}\mathbf{)},\mathbf{L}^{\mathbf{n}(h)}}(\Delta)\right.  ,
\label{product_BM_J}%
\end{align}
which can be verified similarly as in Proposition 5.2.3 of
\cite{pe1992numerical}. By iterative applications of (\ref{product_BM_J}), the
conditional expectation (\ref{int3_6}) can be converted into a linear
combination of the conditional expectations uniformly represented as in
(\ref{int3_5}). Then from the martingale property of stochastic integrals and
the independence between the gamma process and Brownian motion, the
conditional expectation (\ref{int3_5}) equals to zero if there exists some
integer $m\in\left\{  1,2,\ldots,h\right\}  $ such that $i_{m}=1$. Therefore,
the conditional expectation (\ref{int3_2}) can be finally derived as a linear
combination of the terms uniformly represented as
\begin{equation}
z_{1}^{m}\cdot\int_{0}^{\Delta}\int_{0}^{s_{h}}\cdots\int_{0}^{s_{2}}E\left[
\left.  L^{n_{1}}(s_{1})\cdots L^{n_{h-1}}(s_{h-1})L^{n_{h}}(s_{h})\right\vert
L\left(  \Delta\right)  \right]  ds_{1}\cdots ds_{h-1}ds_{h},
\label{cond_Esp_gamma_pre}%
\end{equation}
for some $h\geq1$, $0\leq m\leq h$, $0<s_{1}<s_{2}<\cdots<s_{h}<\Delta$, and
nonnegative integers $n_{1},n_{2},\ldots,n_{h}$. Thus, to calculate
(\ref{int3_2}), it suffices to derive the conditional expectation in
(\ref{cond_Esp_gamma_pre}).

\subsection{Calculating the conditional expectation $E\left[  \left.
L^{n_{1}}(s_{1})L^{n_{2}}(s_{2})\cdots L^{n_{h}}(s_{h})\right\vert L\left(
\Delta\right)  \right]  $}

\label{subsec:Step3}

In this part, we focus on the\ following conditional expectation
\begin{equation}
E\left[  \left.  L^{n_{1}}(s_{1})L^{n_{2}}(s_{2})\cdots L^{n_{h}}%
(s_{h})\right\vert L\left(  \Delta\right)  \right]  \label{cond_Esp_gamma}%
\end{equation}
appeared in (\ref{cond_Esp_gamma_pre}), for some $0<s_{1}<s_{2}<\cdots
<s_{h}<\Delta$ and nonnegative integers $n_{1},n_{2},\ldots,n_{h}$. The
expectation (\ref{cond_Esp_gamma}) involves the product of values of the gamma
process $L\left(  \cdot\right)  $ evaluated at different intermediate times
conditional on the value of $L\left(  \cdot\right)  $ at the terminal time
$\Delta$, and can be represented as a function of $L\left(  \Delta\right)  $
by the following theorem.

\begin{theorem}
For $h\geq1$, $0<s_{1}<s_{2}<\cdots<s_{h}<\Delta$ and nonnegative integers
$n_{1},n_{2},\ldots,n_{h}$, we have
\begin{align}
&  E\left[  \left.  L^{n_{1}}(s_{1})L^{n_{2}}(s_{2})\cdots L^{n_{h}}%
(s_{h})\right\vert L(\Delta)\right] \nonumber\\
=  &  \frac{\prod\limits_{r=0}^{m_{1}-1}\left(  as_{1}+r\right)
\prod\limits_{r=m_{1}}^{m_{2}-1}\left(  as_{2}+r\right)  \cdots\prod
\limits_{r=m_{h-1}}^{m_{h}-1}\left(  as_{h}+r\right)  }{\prod\limits_{r=0}%
^{m_{h}-1}\left(  a\Delta+r\right)  }L^{m_{h}}(\Delta),
\label{key multiplication}%
\end{align}
where $m_{k}=n_{1}+n_{2}+\cdots+n_{k}$ for $k=1,2,\ldots,h$, and the parameter
$a$ is defined through the density function of $L(\cdot)$ in
(\ref{density_gamma process}). \label{prop:CE_gamma bridge}
\end{theorem}

\begin{proof}
We first notice a fact that for the gamma process $L(\cdot)$ with density
function given by (\ref{density_gamma process}) and $t_{0}<t_{1}<t_{2}$,
conditional on $L\left(  t_{0}\right)  =v_{0}$ and $L\left(  t_{2}\right)
=v_{2}$, we have (cf. \cite{ribeiro2004valuing})%
\begin{equation}
L\left(  t_{1}\right)  \overset{d}{=}v_{0}+p\left(  v_{2}-v_{0}\right)  ,
\label{gamma bridge}%
\end{equation}
where $p$ is a random variable following Beta distribution as $p\sim
\mathcal{B}\left(  a\left(  t_{1}-t_{0}\right)  ,a\left(  t_{2}-t_{1}\right)
\right)  $.

Now we return to the proof of this lemma. For $0<s_{1}<s_{2}<\cdots
<s_{h}<\Delta$, by the property of iterated expectation and $L(0)=0$, we can
get%
\begin{align}
&  E\left[  \left.  L^{n_{1}}(s_{1})L^{n_{2}}(s_{2})\cdots L^{n_{h}}%
(s_{h})\right\vert L(\Delta)\right] \nonumber\\
=  &  E\left\{  \left.  E\left[  \left.  L^{n_{1}}(s_{1})L^{n_{2}}%
(s_{2})\cdots L^{n_{h}}(s_{h})\right\vert L(s_{2}),\ldots,L(s_{h}%
),L(\Delta)\right]  \right\vert L(\Delta)\right\} \nonumber\\
=  &  E\left\{  \left.  L^{n_{2}}(s_{2})\cdots L^{n_{h}}(s_{h})E\left[
\left.  L^{n_{1}}(s_{1})\right\vert L(s_{2}),\ldots,L(s_{h}),L(\Delta)\right]
\right\vert L(\Delta)\right\} \nonumber\\
=  &  E\left\{  \left.  L^{n_{2}}(s_{2})\cdots L^{n_{h}}(s_{h})E\left[
\left.  L^{n_{1}}(s_{1})\right\vert L(s_{2})\right]  \right\vert
L(\Delta)\right\}  , \label{int3_7}%
\end{align}
where the last equality follows from the harness property of general L\'{e}vy
process (see, for example, in Section 11.2.7 of \cite{Jeanblanc09}). To
calculate $E\left[  \left.  L^{n_{1}}(s_{1})\right\vert L(s_{2})\right]  $ in
the last line of (\ref{int3_7}), we see from (\ref{gamma bridge}) that given
$L(s_{2})$,%
\[
L(s_{1})\overset{d}{=}p_{1}L(s_{2}),\text{ where }p_{1}\sim\mathcal{B}\left(
as_{1},a\left(  s_{2}-s_{1}\right)  \right)  ,
\]
from which $E\left[  \left.  L^{n_{1}}(s_{1})\right\vert L(s_{2})\right]
=L^{n_{1}}(s_{2})E\left[  p_{1}^{n_{1}}\right]  $ and (\ref{int3_7}) can be
further calculated as%
\begin{equation}
E\left[  \left.  L^{n_{1}}(s_{1})L^{n_{2}}(s_{2})\cdots L^{n_{h}}%
(s_{h})\right\vert L(\Delta)\right]  =E\left[  p_{1}^{n_{1}}\right]  E\left[
\left.  L^{n_{1}+n_{2}}(s_{2})L^{n_{3}}(s_{3})\cdots L^{n_{h}}(s_{h}%
)\right\vert L(\Delta)\right]  . \label{int3_8}%
\end{equation}

Similarly, in the right-hand side of (\ref{int3_8}), we notice that
\begin{align*}
&  \text{ \ \ \ }E\left[  \left.  L^{n_{1}+n_{2}}(s_{2})L^{n_{3}}(s_{3})\cdots
L^{n_{h}}(s_{h})\right\vert L(\Delta)\right]  \\
&  =E\left\{  \left.  L^{n_{3}}(s_{3})\cdots L^{n_{h}}(s_{h})E\left[  \left.
L^{n_{1}+n_{2}}(s_{2})\right\vert L(s_{3})\right]  \right\vert L(\Delta
)\right\}  \\
&  =E\left[  p_{2}^{n_{1}+n_{2}}\right]  E\left[  \left.  L^{n_{1}+n_{2}%
+n_{3}}(s_{3})L^{n_{4}}(s_{4})\cdots L^{n_{h}}(s_{h})\right\vert
L(\Delta)\right]  ,
\end{align*}
where $p_{2}\sim\mathcal{B}\left(  as_{2},a\left(  s_{3}-s_{2}\right)
\right)  $, so that%
\[
E\left[  \left.  L^{n_{1}}(s_{1})L^{n_{2}}(s_{2})\cdots L^{n_{h}}%
(s_{h})\right\vert L(\Delta)\right]  =E\left[  p_{1}^{n_{1}}\right]  E\left[
p_{2}^{n_{1}+n_{2}}\right]  E\left[  \left.  L^{n_{1}+n_{2}+n_{3}}%
(s_{3})L^{n_{4}}(s_{4})\cdots L^{n_{h}}(s_{h})\right\vert L(\Delta)\right]  .
\]
Continuing the above procedure in a similar manner, for $h\geq2$, we deduce
that%
\begin{align}
&  E\left[  \left.  L^{n_{1}}(s_{1})L^{n_{2}}(s_{2})\cdots L^{n_{h}}%
(s_{h})\right\vert L(\Delta)\right]  \nonumber\\
= &  E\left[  p_{1}^{n_{1}}\right]  E\left[  p_{2}^{n_{1}+n_{2}}\right]
\cdots E\left[  p_{h}^{n_{1}+n_{2}+\cdots+n_{h}}\right]  L^{n_{1}+n_{2}%
+\cdots+n_{h}}(\Delta)\nonumber\\
\triangleq &  E\left[  p_{1}^{m_{1}}\right]  E\left[  p_{2}^{m_{2}}\right]
\cdots E\left[  p_{h}^{m_{h}}\right]  L^{m_{h}}(\Delta),\label{int3_9}%
\end{align}
where
\begin{equation}
p_{k}\sim\mathcal{B}\left(  as_{k},a\left(  s_{k+1}-s_{k}\right)  \right)
,\text{ for }1\leq k\leq h-1\label{p_k}%
\end{equation}
and
\begin{equation}
p_{h}\sim\mathcal{B}\left(  as_{h},a\left(  \Delta-s_{h}\right)  \right)
,\label{p_h}%
\end{equation}
with $m_{k}=n_{1}+n_{2}+\cdots+n_{k}$ for $1\leq k\leq h$.

To evaluate the expectation $E\left[  p_{k}^{m_{k}}\right]  $ for
$k=1,2,\ldots,h$ in (\ref{int3_9}),\ we notice that for a random variable
$X\sim\mathcal{B}\left(  \alpha,\beta\right)  $,%
\[
E\left[  X^{k}\right]  =\frac{\alpha^{\left(  k\right)  }}{\left(
\alpha+\beta\right)  ^{\left(  k\right)  }}:=\prod\nolimits_{r=0}^{k-1}%
\frac{\alpha+r}{\alpha+\beta+r}%
\]
holds for any positive integer $k$. Then it follows from (\ref{p_k}) and
(\ref{p_h}) that
\[
E\left[  p_{k}^{m_{k}}\right]  =\prod\nolimits_{r=0}^{m_{k}-1}\frac{as_{k}%
+r}{as_{k+1}+r}=\frac{\prod\nolimits_{r=0}^{m_{k}-1}\left(  as_{k}+r\right)
}{\prod\nolimits_{r=0}^{m_{k}-1}\left(  as_{k+1}+r\right)  },\text{ for }1\leq
k\leq h-1,
\]
and%
\[
E\left[  p_{h}^{m_{h}}\right]  =\prod\nolimits_{r=0}^{m_{h}-1}\frac{as_{h}%
+r}{a\Delta+r}=\frac{\prod\nolimits_{r=0}^{m_{h}-1}\left(  as_{h}+r\right)
}{\prod\nolimits_{r=0}^{m_{h}-1}\left(  a\Delta+r\right)  }.
\]
Plugging the above two equations into (\ref{int3_9}), we obtain that%
\begin{align*}
&  \text{ \ \ \ }E\left[  \left.  L^{n_{1}}(s_{1})L^{n_{2}}\left(
s_{2}\right)  \cdots L^{n_{h}}\left(  s_{h}\right)  \right\vert L\left(
\Delta\right)  \right] \\
&  =\frac{\prod\nolimits_{r=0}^{m_{1}-1}\left(  as_{1}+r\right)  }%
{\prod\nolimits_{r=0}^{m_{1}-1}\left(  as_{2}+r\right)  }\frac{\prod
\nolimits_{r=0}^{m_{2}-1}\left(  as_{2}+r\right)  }{\prod\nolimits_{r=0}%
^{m_{2}-1}\left(  as_{3}+r\right)  }\cdots\frac{\prod\nolimits_{r=0}%
^{m_{h-1}-1}\left(  as_{h-1}+r\right)  }{\prod\nolimits_{r=0}^{m_{h-1}%
-1}\left(  as_{h}+r\right)  }\frac{\prod\nolimits_{r=0}^{m_{h}-1}\left(
as_{h}+r\right)  }{\prod\nolimits_{r=0}^{m_{h}-1}\left(  a\Delta+r\right)
}L^{m_{h}}\left(  \Delta\right) \\
&  =\frac{\prod\nolimits_{r=0}^{m_{1}-1}\left(  as_{1}+r\right)
\prod\nolimits_{r=m_{1}}^{m_{2}-1}\left(  as_{2}+r\right)  \cdots
\prod\nolimits_{r=m_{h-1}}^{m_{h}-1}\left(  as_{h}+r\right)  }{\prod
\nolimits_{r=0}^{m_{h}-1}\left(  a\Delta+r\right)  }L^{m_{h}}\left(
\Delta\right)  ,
\end{align*}
which concludes the proof.
\end{proof}

Based on the previous calculations, the function $K_{\left(  \ell
,\mathbf{j}\left(  \ell\right)  \right)  }\left(  z_{1},z_{2}\right)  $
defined by (\ref{K(z)}) can be expressed as a linear combination of
$z_{1}^{n_{1}}z_{2}^{n_{2}}$ for $n_{1},n_{2}\geq0$ with the coefficients
depending on functions $\mu\left(  x_{0}\right)  $, $\sigma(x_{0})$ and
their\ higher-order derivatives evaluated at $x_{0}$. According to the
definition of the partial differential operators $\mathcal{D}_{1}^{(\ell
)}\left(  \cdot\right)  $ for $\ell\geq1$ in (\ref{itera derivative}), the
expression of $\mathcal{D}_{1}^{(\ell)}\left(  K_{\left(  \ell,\mathbf{j}%
\left(  \ell\right)  \right)  }\left(  z_{1},z_{2}\right)  \right)  $ can also
be established as a linear combination of $z_{1}^{n_{1}}z_{2}^{n_{2}}$ for
$n_{1},n_{2}\geq0$, denoted by%
\[
\mathcal{D}_{1}^{(\ell)}\left(  K_{\left(  \ell,\mathbf{j}\left(  \ell\right)
\right)  }\left(  z_{1},z_{2}\right)  \right)  =%
{\textstyle\sum\limits_{n_{1},n_{2}\geq0}}
P_{\left(  \ell,\mathbf{j}\left(  \ell\right)  \right)  }^{\mu,\sigma}\left(
n_{1},n_{2}\right)  z_{1}^{n_{1}}z_{2}^{n_{2}}%
\]
for some coefficient functions $P_{\left(  \ell,\mathbf{j}\left(  \ell\right)
\right)  }^{\mu,\sigma}\left(  n_{1},n_{2}\right)  $ and $n_{1},n_{2}\geq0$.
Then it follows from (\ref{Omega_m(y)}) -- (\ref{z1}) that $\Omega_{m}\left(
y\right)  $ can be finally represented as
\[
\Omega_{m}(y)=\sum_{\left(  \ell,\left(  j_{1},j_{2},\cdots,j_{\ell}\right)
\right)  \in\mathcal{S}_{m}}\frac{(-1)^{^{\ell}}}{\ell!}\frac{1}{(\sigma
(x_{0})\sqrt{\Delta})^{\ell}}%
{\textstyle\sum\limits_{n_{1},n_{2}\geq0}}
P_{\left(  \ell,\mathbf{j}\left(  \ell\right)  \right)  }^{\mu,\sigma}\left(
n_{1},n_{2}\right)  \int_{0}^{+\infty}z_{1}^{n_{1}}z_{2}^{n_{2}}\phi
(z_{1})p_{L\left(  \Delta\right)  }(z_{2})dz_{2},
\]
where $z_{1}=y-\frac{\mu\left(  x_{0}\right)  \Delta+z_{2}}{\sigma(x_{0}%
)\sqrt{\Delta}}$ and $y=\frac{x-x_{0}}{\sigma(x_{0})\sqrt{\Delta}}$, the index
set $\mathcal{S}_{m}$ is defined in (\ref{Sm}), $\mu\left(  \cdot\right)  $
and $\sigma\left(  \cdot\right)  $ are defined through the SDE (\ref{model}),
$\phi(\cdot)$ is the standard normal density function\ and $p_{L\left(
\Delta\right)  }(\cdot)$ is the density function of $L(\Delta)$ in
(\ref{density_gamma process}).

\begin{remark}
For the special case of the SDE (\ref{model}) with $\sigma
(X(t);\bm{\theta})\equiv0$ in Remark \ref{remark:PJ}, to calculate $\Omega
_{m}(y)$ for $m=0,1,2,...$ in (\ref{pX_M_expan_PJ}), we skip the Step 2 in the
general algorithm stated prior to Section \ref{subsec:Step1}. The remaining
procedures are performed in a similar manner.
\end{remark}

\subsection{Examples}

\label{subsec:examples}

In this section, we consider the pure jump OU model, the constant diffusion
model and the square-root diffusion model as three examples of SDE
(\ref{model}) to give the concrete expressions of the first several expansion
terms of $\left\{  \Omega_{m}(y),m\geq0\right\}  $ in (\ref{pX_M_expan}).

The first model below is the pure jump OU process which is a special case of
the non-Gaussian OU processes proposed by \cite{barndorff2001non}. The pure
jump OU process is widely used in finance and statistical analysis, e.g., to
specify the stochastic volatility driving the dynamics of asset prices. We
refer to \cite{barndorff2001non}, \cite{schoutens03} and \cite{ContTankov04}
for more details of the non-Gaussian OU processes.

Model 1 (Pure jump OU model). By taking $\bm{\theta}=\left\{  \kappa
,\theta\right\}  $ and\ letting $\mu\left(  x;\bm{\theta}\right)
=\kappa\left(  \theta-x\right)  $ and $\sigma\left(  x;\bm{\theta}\right)
\equiv0$ in (\ref{model}), we obtain the pure jump model
\begin{equation}
dX(t)=\kappa\left(  \theta-X(t)\right)  dt+dL(t),\text{ }X\left(  0\right)
=x_{0}. \label{PJ_OU}%
\end{equation}
The function $\Omega_{m}\left(  y\right)  $ for $m=0,1,2,3$ in
(\ref{pX_M_expan_PJ}) can be calculated as
\begin{align*}
\Omega_{0}\left(  y\right)   &  =\frac{b^{a\Delta}\left(  y-\eta\Delta\right)
^{a\Delta-1}e^{-b\left(  y-\eta\Delta\right)  }}{\Gamma\left(  a\Delta\right)
},\\
\Omega_{1}\left(  y\right)   &  =-\frac{b^{a\Delta}\left(  y-\eta
\Delta\right)  ^{a\Delta-2}e^{-b\left(  y-\eta\Delta\right)  }}{2\Gamma\left(
a\Delta\right)  }\kappa\Delta\left[  \left(  by-a\Delta\right)  y+\eta
\Delta\left(  1-by\right)  \right]  ,\\
\Omega_{2}\left(  y\right)   &  =\frac{b^{a\Delta}\left(  y-\eta\Delta\right)
^{a\Delta-3}e^{-b\left(  y-\eta\Delta\right)  }}{24\left(  1+a\Delta\right)
\Gamma\left(  a\Delta\right)  }\kappa^{2}\Delta^{2}\left[  b^{2}\left(
y-\eta\Delta\right)  ^{2}\left(  y^{2}\left(  4+3a\Delta\right)  -2y\eta
\Delta+\eta^{2}\Delta^{2}\right)  \right. \\
&  \left.  \text{ \ \ }-2b\left(  1+a\Delta\right)  \left(  y-\eta
\Delta\right)  \left(  y^{2}\left(  2+3a\Delta\right)  -6\eta\Delta
y+\kappa^{2}\eta^{2}\right)  \right. \\
&  \left.  \text{ \ \ }+\left(  1+a\Delta\right)  \left(  3a^{2}y^{2}%
\Delta^{2}+2\left(  2-5a\Delta\right)  \eta\Delta y+\left(  2+a\Delta\right)
\eta^{2}\Delta^{2}\right)  \right]  ,
\end{align*}
and%
\begin{align*}
\Omega_{3}\left(  y\right)   &  =-\frac{b^{a\Delta}\left(  y-\eta
\Delta\right)  ^{a\Delta-4}e^{-b\left(  y-\eta\Delta\right)  }}{48\left(
1+a\Delta\right)  \Gamma\left(  a\Delta\right)  }\kappa^{3}\Delta^{3}%
\times\left\{  \left[  b^{3}\left(  2+a\Delta\right)  y^{3}-b^{2}\left(
6+a\Delta\left(  8+3a\Delta\right)  \right)  y^{2}\right.  \right. \\
&  \left.  \text{ \ }+b\left(  1+a\Delta\right)  \left(  2+a\Delta\left(
4+3a\Delta\right)  \right)  y-a^{3}\Delta^{3}\left(  1+a\Delta\right)
\right]  y^{3}\\
&  \left.  \text{ \ }-\eta\Delta\left[  b^{3}\left(  8+3a\Delta\right)
y^{3}-b^{2}\left(  26+3a\Delta\left(  9+2a\Delta\right)  \right)
y^{2}+b\left(  1+a\Delta\right)  \left(  6+a\Delta\left(  20+3a\Delta\right)
\right)  y\right.  \right. \\
&  \left.  \left.  \text{ \ }-\left(  1+a\Delta\right)  \left(  2+a\Delta
\left(  -6+7a\Delta\right)  \right)  \right]  y^{2}+\eta^{2}\Delta^{2}\left[
b^{3}\left(  13+3a\Delta\right)  y^{3}-b^{2}\left(  40+3a\Delta\left(
11+a\Delta\right)  \right)  y^{2}\right.  \right. \\
&  \left.  \left.  \text{ \ }+b\left(  1+a\Delta\right)  \left(
14+19a\Delta\right)  y-\left(  1+a\Delta\right)  \left(  -4+a\Delta\left(
6+a\Delta\right)  \right)  \right]  y\right. \\
&  \left.  \text{ \ }-\eta^{3}\Delta^{3}\left[  b^{3}\left(  11+a\Delta
\right)  y^{3}-b^{2}\left(  27+17a\Delta\right)  y^{2}+3b\left(
1+a\Delta\right)  \left(  4+a\Delta\right)  y-a\Delta\left(  1+a\Delta\right)
\right]  \right. \\
&  \left.  \text{ \ }+b\eta^{4}\Delta^{4}\left[  by\left(  -8-3a\Delta
+5by\right)  +2\left(  1+a\Delta\right)  \right]  +b^{2}\eta^{5}\Delta
^{5}\left(  1-by\right)  \right\}  ,
\end{align*}
where $\eta:=\kappa\left(  \theta-x_{0}\right)  $.

The following two models generalize the pure jump OU process (\ref{PJ_OU}) in
Model 1 with extra innovation driven by the Brownian motion, specified as the
constant diffusion and square-root diffusion respectively. We refer to
\cite{kunita2019stochastic} for more advanced descriptions of the
jump-diffusion SDEs driven by general L\'{e}vy processes.

Model 2 (Constant diffusion model). By taking $\bm{\theta}=\left\{
\kappa,\theta,\sigma\right\}  $ and\ letting $\mu\left(  x;\bm{\theta}\right)
=\kappa\left(  \theta-x\right)  $ and $\sigma\left(  x;\bm{\theta}\right)
\equiv\sigma>0$ in (\ref{model}), we obtain the constant diffusion model%

\begin{equation}
dX(t)=\kappa\left(  \theta-X(t)\right)  dt+\sigma dW(t)+dL(t),\text{ }X\left(
0\right)  =x_{0}. \label{cons_diffu_OU}%
\end{equation}
The function $\Omega_{m}(y)$ for $m=0,1,2,3$ in (\ref{pX_M_expan}) can be
calculated as%
\begin{align*}
&  \left.  \Omega_{0}(y)=S_{0}(y),\right. \\
&  \left.  \Omega_{1}(y)=\frac{\kappa\Delta^{1/2}}{2\sigma}\left\{
yS_{1}(y)+\left[  \eta\Delta y+\sigma\Delta^{1/2}\left(  1-y^{2}\right)
\right]  S_{0}(y)\right\}  ,\right. \\
&  \left.  \Omega_{2}(y)=\frac{\kappa^{2}}{24\sigma^{4}\left(  1+a\Delta
\right)  }\right. \\
&  \text{ \ \ \ \ \ \ \ \ \ \ }\times\left\{  S_{4}(y)+2\left(  \eta
\Delta-\sigma\Delta^{1/2}y\right)  S_{3}(y)+\left[  \eta^{2}\Delta-2\eta
\sigma\Delta^{1/2}y+\sigma^{2}\left(  a\Delta+\left(  4+3a\Delta\right)
y^{2}\right)  \right]  \Delta S_{2}(y)\right. \\
&  \left.  \text{ \ \ \ \ \ \ \ \ \ }+2\sigma^{2}\Delta^{3/2}\left(
1+a\Delta\right)  \left[  \eta\left(  1+3y^{2}\right)  \Delta^{1/2}%
+3\sigma\left(  1-y^{2}\right)  y\right]  S_{1}(y)\right. \\
&  \left.  \text{ \ \ \ \ \ \ \ \ \ \ \ \ }\sigma^{2}\Delta^{2}\left(
1+a\Delta\right)  \left[  \eta^{2}\left(  1+3y^{2}\right)  \Delta+6\eta
\sigma\Delta^{1/2}\left(  1-y^{2}\right)  y+\sigma^{2}\left(  1-10y^{2}%
+3y^{4}\right)  \right]  S_{0}(y)\right\}  ,
\end{align*}
and%
\begin{align*}
&  \Omega_{3}(y)=\frac{\kappa^{3}\Delta^{1/2}}{336\sigma^{5}\left(
1+a\Delta\right)  }\times\left\{  7yS_{5}(y)+\left[  21\eta\Delta
y-\sigma\Delta^{1/2}\left(  -4+3a\Delta+21y^{2}\right)  \right]
S_{4}(y)\right. \\
&  \left.  \text{ \ \ \ \ \ \ }+\left[  21\eta^{2}\Delta y+\eta\sigma
\Delta^{1/2}\left(  5-9a\Delta-42y^{2}\right)  +\sigma^{2}\left(
-19+16a\Delta+7\left(  4+a\Delta\right)  y^{2}\right)  y\right]  \Delta
S_{3}(y)\right. \\
&  \left.  \text{ \ \ \ \ \ \ }+\left[  7\eta^{3}\Delta^{3/2}y-\eta^{2}%
\sigma\Delta\left(  2+9a\Delta+21y^{2}\right)  +\eta\sigma^{2}\Delta
^{1/2}y\left(  4+39a\Delta+21\left(  2+a\Delta\right)  y^{2}\right)  \right.
\right. \\
&  \left.  \left.  \text{ \ \ \ \ \ \ }+\sigma^{3}\left(  9+16a\Delta
-7(4+3a\Delta)y^{4}+\left(  33+5a\Delta\right)  y^{2}\right)  \right]
\Delta^{3/2}S_{2}(y)\right. \\
&  \left.  \text{ \ \ \ \ \ \ }+\sigma\Delta^{2}\left(  1+a\Delta\right)
\left[  -3\eta^{3}\Delta^{3/2}+3\eta^{2}\sigma\Delta\left(  10+7y^{2}\right)
y+\eta\sigma^{2}\Delta^{1/2}\left(  23+19y^{2}-42y^{4}\right)  \right.
\right. \\
&  \left.  \left.  \text{ \ \ \ \ \ \ }+\sigma^{3}\left(  -16-67y^{2}%
+21y^{4}\right)  y\right]  S_{1}(y)+7\sigma^{2}\Delta^{5/2}\left(
1+a\Delta\right)  \right. \\
&  \text{ \ \ \ \ \ \ \ }\times\left[  \eta^{3}\Delta^{3/2}\left(
1+y^{2}\right)  y+\eta^{2}\sigma\Delta\left(  1+2y^{2}-3y^{4}\right)  \right.
\\
&  \left.  \left.  \text{ \ \ \ \ \ \ }+\eta\sigma^{2}\Delta^{1/2}\left(
-1-10y^{2}+3y^{4}\right)  y-\sigma^{3}\left(  1+5y^{2}-7y^{4}+y^{6}\right)
\right]  S_{0}(y)\right\}  ,
\end{align*}
where $\eta:=\kappa\left(  \theta-x_{0}\right)  $ and for $m=0,1,2,\ldots$, \
\begin{align*}
&  S_{m}(y):=\frac{1}{2\sqrt{2\pi}}\frac{b^{a\Delta}}{\Gamma\left(
a\Delta\right)  }A^{-1-r_{m}/2}\exp\left[  -\frac{1}{2}\left(  y-\kappa\left(
\theta-x_{0}\right)  \Delta^{1/2}/\sigma\right)  ^{2}\right] \\
&  \text{ \ \ \ \ \ \ \ \ \ \ \ }\times\left[  B\Gamma\left(  1+\frac{r_{m}%
}{2}\right)  {}_{1}F_{1}\left(  1+\frac{r_{m}}{2},\frac{3}{2};\frac{B^{2}}%
{4A}\right)  +\sqrt{A}\Gamma\left(  \frac{1+r_{m}}{2}\right)  {}_{1}%
F_{1}\left(  \frac{1+r_{m}}{2},\frac{1}{2};\frac{B^{2}}{4A}\right)  \right]
\end{align*}
with $r_{m}:=m+a\Delta-1$, $A:=1/\left(  2\sigma^{2}\Delta\right)  $,
$B:=\frac{y}{\sigma\Delta^{1/2}}-\frac{\kappa\left(  \theta-x_{0}\right)
}{\sigma^{2}}-b$, and
\[
_{1}F_{1}\left(  a,b;z\right)  :=\sum\nolimits_{k=0}^{\infty}\left(
(a)_{k}/(b)_{k}\right)  z^{k}/k!
\]
as the Kummer confluent hypergeometric function.

Model 3 (Square-root diffusion model). By taking $\bm{\theta}=\left\{
\kappa,\theta,\sigma\right\}  $ and\ letting $\mu\left(  x;\bm{\theta}\right)
=\kappa\left(  \theta-x\right)  $ and $\sigma\left(  x;\bm{\theta}\right)
\equiv\sigma\sqrt{x}$ with $\sigma>0$\ in (\ref{model}), we obtain the
square-root diffusion model
\begin{equation}
dX(t)=\kappa\left(  \theta-X(t)\right)  dt+\sigma\sqrt{X(t)}dW(t)+dL(t),\text{
}X\left(  0\right)  =x_{0}. \label{SQRT_diffu_OU}%
\end{equation}
The function $\Omega_{m}(y)$ for $m=0,1,2$ in (\ref{pX_M_expan}) can be
calculated as%
\begin{align*}
\Omega_{0}(y)  &  =S_{0}(y),\\
\Omega_{1}(y)  &  =\frac{1}{4\sigma x_{0}^{2}\sqrt{\Delta}}\left\{
yS_{2}(y)+\left[  2\kappa\theta\Delta y+2\sigma\left(  x_{0}-y^{2}\right)
\sqrt{\Delta}\right]  S_{1}(y)\right. \\
&  \left.  \text{ \ \ }\left[  \kappa\eta\left(  \theta+x_{0}\right)
\Delta^{2}y+2\kappa\theta\sigma\left(  x_{0}-y^{2}\right)  \Delta^{3/2}%
+\sigma^{2}\left(  -3x_{0}+y^{2}\right)  \Delta y\right]  S_{0}(y)\right\}  ,
\end{align*}
and%
\begin{align*}
&  \Omega_{2}(y)=\frac{1}{480\sigma^{4}x_{0}^{4}\Delta^{2}\left(
1+a\Delta\right)  }\left\{  5S_{6}(y)+20\left(  \kappa\theta\Delta-\sigma
\sqrt{\Delta}y\right)  S_{5}(y)\right. \\
&  \text{ \ \ \ \ \ \ \ \ \ }+\left[  10\kappa^{2}\left(  3\theta^{2}%
-x_{0}^{2}\right)  \Delta^{2}-60\kappa\theta\sigma\Delta^{3/2}y+\left(
15\left(  3+a\Delta\right)  y^{2}+\left(  -12+17a\Delta\right)  x_{0}\right)
\sigma^{2}\Delta\right]  S_{4}(y)\\
&  \text{ \ \ \ \ \ \ \ \ \ }+20\kappa^{3}\theta\left(  \theta^{2}-x_{0}%
^{2}\right)  \Delta^{3}-20\kappa^{2}\sigma\left(  3\theta^{2}-x_{0}%
^{2}\right)  \Delta^{5/2}y\\
&  \text{ \ \ \ \ \ \ \ \ \ }-2\kappa\sigma^{2}\left(  \left(  7+6a\Delta
\right)  x_{0}^{2}+\theta\left(  x_{0}\left(  1-28a\Delta\right)  -30\left(
2+a\Delta\right)  y^{2}\right)  \right)  \Delta^{2}\\
&  \left.  \text{ \ \ \ \ \ \ \ \ }+2\sigma^{3}\left(  x_{0}\left(
38+9a\Delta\right)  -10\left(  4+3a\Delta\right)  y^{2}\right)  \Delta
^{3/2}y\right]  S_{3}(y)\\
&  \text{ \ \ \ \ \ \ \ \ \ }+\left[  5\kappa^{4}\left(  \theta^{2}-x_{0}%
^{2}\right)  ^{2}\Delta^{2}-20\kappa^{3}\sigma\theta\left(  \theta^{2}%
-x_{0}^{2}\right)  \Delta^{3/2}y\right. \\
&  \left.  \text{ \ \ \ \ \ \ \ \ }+\left(  \eta^{2}\left(  30y^{2}\left(
4+3a\Delta\right)  +x_{0}\left(  37+66a\Delta\right)  \right)  +12\kappa\eta
x_{0}\left(  5y^{2}\left(  4+3a\Delta\right)  +x_{0}\left(  4+9a\Delta\right)
\right)  \right.  \right. \\
&  \text{ \ \ \ \ \ \ \ \ \ }+\left.  20\kappa^{2}x_{0}^{2}\left(
x_{0}a\Delta+\left(  4+3a\Delta\right)  y^{2}\right)  \right)  \sigma
^{2}\Delta+2\sigma^{3}\eta\left(  -10y^{2}\left(  10+9a\Delta\right)
+x_{0}\left(  77+48a\Delta\right)  \right)  \sqrt{\Delta}y\\
&  \text{ \ \ \ \ \ \ \ \ \ }+4\kappa\sigma^{3}x_{0}\left(  -5y^{2}\left(
10+9a\Delta\right)  +x_{0}\left(  48+33a\Delta\right)  \right)  \sqrt{\Delta
}y\\
&  \text{ \ \ \ \ \ \ \ \ \ }-\sigma^{4}\left(  -5\left(  19+18a\Delta\right)
y^{4}+\left(  281+252a\Delta\right)  x_{0}y^{2}\right. \\
&  \left.  \left.  \text{ \ \ \ \ \ \ \ \ }+\left(  9+23a\Delta\right)
x_{0}^{2}\right)  \right]  \Delta^{2}S_{2}(y)+\sigma^{2}\left(  1+a\Delta
\right)  \Delta^{5/2}\left[  -4\kappa^{3}\left(  \theta^{2}-x_{0}^{2}\right)
\left(  3x_{0}^{2}-8\theta x_{0}-15\theta y^{2}\right)  \Delta^{3/2}\right. \\
&  \text{ \ \ \ \ \ \ \ \ \ }+6\sigma\left(  \eta^{2}\left(  23x_{0}%
-30y^{2}\right)  +4\kappa\eta x_{0}\left(  13x_{0}-15y^{2}\right)
+20\kappa^{2}x_{0}^{2}\left(  x_{0}-y^{2}\right)  \right)  \Delta y\\
&  \text{ \ \ \ \ \ \ \ \ \ }+\sigma^{2}\left(  \eta\left(  7x_{0}%
^{2}-552x_{0}y^{2}+180y^{4}\right)  +36\kappa x_{0}\left(  x_{0}^{2}%
-16x_{0}y^{2}+5y^{4}\right)  \right)  \sqrt{\Delta}\\
&  \left.  \text{ \ \ \ \ \ \ \ \ }+\sigma^{3}\left(  -241x_{0}^{2}%
+382x_{0}y^{2}-60y^{4}\right)  y\right]  S_{1}(y)\\
&  \left.  \text{ \ \ \ \ \ \ \ \ }+5\sigma^{2}\left(  1+a\Delta\right)
\Delta^{3}\kappa^{2}\eta^{2}\left(  x_{0}+3y^{2}\right)  \left(  \theta
+x_{0}\right)  ^{2}\Delta^{2}+12\kappa^{2}\theta\sigma\eta\left(  x_{0}%
-y^{2}\right)  \left(  \theta+x_{0}\right)  \Delta^{3/2}y\right. \\
&  \left.  \text{ \ \ \ \ \ \ \ \ }+2\sigma^{2}\Delta\left(  x_{0}^{2}%
-10x_{0}y^{2}+3y^{4}\right)  \left(  3\eta^{2}+6\kappa\eta x_{0}+2\kappa
^{2}x_{0}^{2}\right)  -4\kappa\theta\sigma^{3}\left(  15x_{0}^{2}-20x_{0}%
y^{2}+3y^{4}\right)  \sqrt{\Delta}y\right. \\
&  \left.  \left.  \text{ \ \ \ \ \ \ \ \ }+3\sigma^{4}\left(  -3x_{0}%
^{3}+21x_{0}^{2}y^{2}-11x_{0}y^{4}+y^{6}\right)  \right]  S_{0}(y)\right\}  ,
\end{align*}
where $\eta:=\kappa\left(  \theta-x_{0}\right)  $ and for $m=0,1,2,\ldots$,%
\begin{align*}
&  S_{m}(y):=\frac{1}{2\sqrt{2\pi}}\frac{b^{a\Delta}}{\Gamma\left(
a\Delta\right)  }A^{-1-r_{m}/2}\exp\left[  -\frac{1}{2}\left(  y-\frac
{\kappa\left(  \theta-x_{0}\right)  \Delta^{1/2}}{\sigma\sqrt{x_{0}}}\right)
^{2}\right] \\
&  \text{ \ \ \ \ \ \ \ \ \ \ \ }\times\left[  B\Gamma\left(  1+\frac{r_{m}%
}{2}\right)  {}_{1}F_{1}\left(  1+\frac{r_{m}}{2},\frac{3}{2};\frac{B^{2}}%
{4A}\right)  +\sqrt{A}\Gamma\left(  \frac{1+r_{m}}{2}\right)  {}_{1}%
F_{1}\left(  \frac{1+r_{m}}{2},\frac{1}{2};\frac{B^{2}}{4A}\right)  \right]
\end{align*}
with $r_{m}:=m+a\Delta-1$, $A:=1/\left(  2\sigma^{2}x_{0}\Delta\right)  $, and
$B:=\frac{y}{\sigma\sqrt{x_{0}}\Delta^{1/2}}-\frac{\kappa\left(  \theta
-x_{0}\right)  }{\sigma^{2}x_{0}}-b$.

\section{Numerical performance}

\label{sec:numerical} In this section, we demonstrate the performance of the
approximations for transition densities via the previous introduced pure jump
OU model, constant diffusion model and square-root diffusion model in Section
\ref{subsec:examples} as examples. To test the accuracy of the asymptotic
expansion methodology, we calculate the true transition density by inverse
Fourier transform of its known characteristic function as the benchmark for
each of the above three models. Here, we use the numerical inverse Fourier
transform method proposed by \cite{Abate_Whitt}, which has been proved to be
efficient and accurate. By comparing the approximated transition densities
using our asymptotic expansion method with the true densities obtained by
inverse Fourier transform, we show that the approximation errors decrease
quickly as the approximation order $M$ in (\ref{pX_M_expan}) increases.

For each example, the true transition density of $X(\Delta)$ can be obtained
via its characteristic function $\phi\left(  \Delta;\omega\right)
:=Ee^{i\omega X(\Delta)}$ by%
\begin{align}
p_{X\left(  \Delta\right)  }\left(  x|x_{0};\theta\right)   &  =\frac{1}{2\pi
}\int_{-\infty}^{+\infty}e^{-ix\omega}\phi\left(  \Delta;\omega\right)
d\omega\nonumber\\
&  =\frac{1}{\pi}\int_{0}^{+\infty}\left[  \cos\left(  x\omega\right)
\operatorname{Re}\left(  \phi\right)  \left(  \Delta;\omega\right)
+\sin\left(  x\omega\right)  \operatorname{Im}\left(  \phi\right)  \left(
\Delta;\omega\right)  \right]  d\omega, \label{true density}%
\end{align}
which can be efficiently approximated via the following Euler summation as%
\[
E\left(  m,n,x\right)  =\sum\limits_{k=1}^{m}\left(
\begin{array}
[c]{c}%
m\\
k
\end{array}
\right)  2^{-m}s_{n+k}(x),
\]
where the truncated series is defined by%
\[
s_{n}(x):=\frac{h}{2\pi}+\frac{h}{\pi}\sum\limits_{k=1}^{n}\left[
\operatorname{Re}\left(  \phi\right)  \left(  \Delta;kh\right)  \cos\left(
khx\right)  +\operatorname{Im}\left(  \phi\right)  \left(  \Delta;kh\right)
\sin\left(  khx\right)  \right]  .
\]
We refer to \cite{Abate_Whitt} for more technical details.

By using (\ref{characteristic exponent}), we can derive the explicit
expressions of the characteristic functions for the above three models as
follows. For the pure jump OU model, we have%
\[
\phi\left(  t;\omega\right)  =\exp\left\{  i\omega\left[  e^{-\kappa t}%
x_{0}+\theta\left(  1-e^{-\kappa t}\right)  \right]  -\frac{a}{\kappa}\left[
\text{Li}_{2}\left(  \frac{i\omega e^{-\kappa t}}{b}\right)  -\text{Li}%
_{2}\left(  \frac{i\omega}{b}\right)  \right]  \right\}  ,
\]
where Li$_{s}(z):=\sum\nolimits_{k=1}^{+\infty}z^{k}/k^{s}$ is the
polylogarithm function. For the constant diffusion model, we have%
\[
\phi\left(  t;\omega\right)  =\exp\left\{  i\omega\left[  e^{-\kappa t}%
x_{0}+\theta\left(  1-e^{-\kappa t}\right)  \right]  -\frac{\omega^{2}%
\sigma^{2}\left(  1-e^{-2\kappa t}\right)  }{4\kappa}-\frac{a}{\kappa}\left[
\text{Li}_{2}\left(  \frac{i\omega e^{-\kappa t}}{b}\right)  -\text{Li}%
_{2}\left(  \frac{i\omega}{b}\right)  \right]  \right\}  .
\]
For the square-root diffusion model, we have%
\[
\phi\left(  t;\omega\right)  =\mathbb{E}\left[  e^{i\omega X\left(  t\right)
}|X(0)=x_{0}\right]  =e^{\alpha\left(  t\right)  +\beta\left(  t\right)
x_{0}},
\]
where $\beta\left(  t\right)  =\frac{2i\omega\kappa}{2\kappa e^{\kappa
t}+i\omega\sigma^{2}\left(  1-e^{\kappa t}\right)  }$ and%
\begin{align*}
\alpha\left(  t\right)   &  =\frac{1}{\sigma^{2}}\left\{  2\kappa^{2}\theta
t+a\sigma^{2}t\log\left(  b\right)  -\left(  2\kappa\theta+a\sigma
^{2}t\right)  \log\left[  1-e^{\kappa t}\left(  1-\frac{2\kappa}{i\omega
\sigma^{2}}\right)  \right]  +2\kappa\theta\log\left(  \frac{2\kappa}%
{i\omega\sigma^{2}}\right)  \right. \\
&  \left.  \text{ \ \ }+a\sigma^{2}t\left[  \log\left(  1-\frac{be^{\kappa
t}\left(  2\kappa-i\omega\sigma^{2}\right)  }{i\omega\left(  2\kappa
-b\sigma^{2}\right)  }\right)  -\log\left(  b-\frac{2i\omega\kappa}%
{i\omega\sigma^{2}\left(  1-e^{\kappa t}\right)  +2\kappa e^{\kappa t}%
}\right)  \right]  \right\} \\
&  \text{ \ \ \ }+\frac{a}{\kappa}\left[  \text{Li}_{2}\left(  1-\frac
{2\kappa}{i\omega\sigma^{2}}\right)  -\text{Li}_{2}\left(  e^{\kappa t}\left(
1-\frac{2\kappa}{i\omega\sigma^{2}}\right)  \right)  \right. \\
&  \left.  \text{ \ \ }-\text{Li}_{2}\left(  \frac{b\left(  2\kappa
-i\omega\sigma^{2}\right)  }{i\omega\left(  2\kappa-b\sigma^{2}\right)
}\right)  +\text{Li}_{2}\left(  \frac{be^{\kappa t}\left(  2\kappa
-i\omega\sigma^{2}\right)  }{i\omega\left(  2\kappa-b\sigma^{2}\right)
}\right)  \right]  .
\end{align*}

For the numerical comparison, according to \cite{barndorff2001non},
\cite{li2016estimating} and \cite{james2013exact},\ we set the parameters of
the above three examples as follows. For the pure jump OU model in
(\ref{PJ_OU}), we set $\kappa=0.6$, $\theta=0.02$, $a=100$, and $b=10$. For
the constant diffusion model in (\ref{cons_diffu_OU}), we set $\kappa=0.6$,
$\theta=0.02$, $\sigma=0.3$, $a=100$, and $b=10$. For the square-root
diffusion model in (\ref{SQRT_diffu_OU}), we set $\kappa=0.6$, $\theta=0.02$,
$\sigma=0.3$, $a=100$, and $b=10$. In each model, we set the initial value
$x_{0}=0.3$.

For each model, given the true transition density $p_{X\left(  \Delta\right)
}\left(  x|x_{0};\theta\right)  $ in (\ref{true density}) and the approximated
density $p_{X\left(  \Delta\right)  }^{\left(  M\right)  }\left(  \left.
x\right\vert x_{0};\theta\right)  $ derived by (\ref{pX_M_expan}) or
(\ref{pX_M_expan_PJ}), we denote by%
\begin{equation}
e_{M}\left(  \left.  \Delta,x\right\vert x_{0};\theta\right)  =p_{X\left(
\Delta\right)  }\left(  x|x_{0};\theta\right)  -p_{X\left(  \Delta\right)
}^{\left(  M\right)  }\left(  \left.  x\right\vert x_{0};\theta\right)
\label{approximated error}%
\end{equation}
the $M$-th order approximation error and define the maximum relative error as
\[
\max_{x\in\mathcal{D}}\left\vert \frac{e_{M}\left(  \left.  \Delta
,x\right\vert x_{0};\theta\right)  }{p_{X\left(  \Delta\right)  }\left(
x|x_{0};\theta\right)  }\right\vert
\]
over a region $\mathcal{D}$.

We consider monthly, weekly, and daily monitoring frequencies ($\Delta
=1/12,1/52,1/252$) and plot the maximum relative errors of order $M=0,1,2,3$
for each model in Figure \ref{fig:Max_relative_abso_err}. It is easy to
observe that the maximum relative errors decrease quickly as the order of
expansion increases. For example, when we choose $\Delta=1/252$ and the order
of $M=2$, the maximum relative error of each model can attain the level as
$10^{-5}$. Besides, when the monitoring frequency rises, i.e. the time
interval $\Delta$ becomes smaller, the maximum relative error will decrease
correspondingly for each model.%

\begin{figure}[H]%
\begin{center}%
\begin{tabular}
[c]{ccc}%
{\parbox[b]{2.1003in}{\begin{center}
\includegraphics[
height=1.7996in,
width=2.1003in
]%
{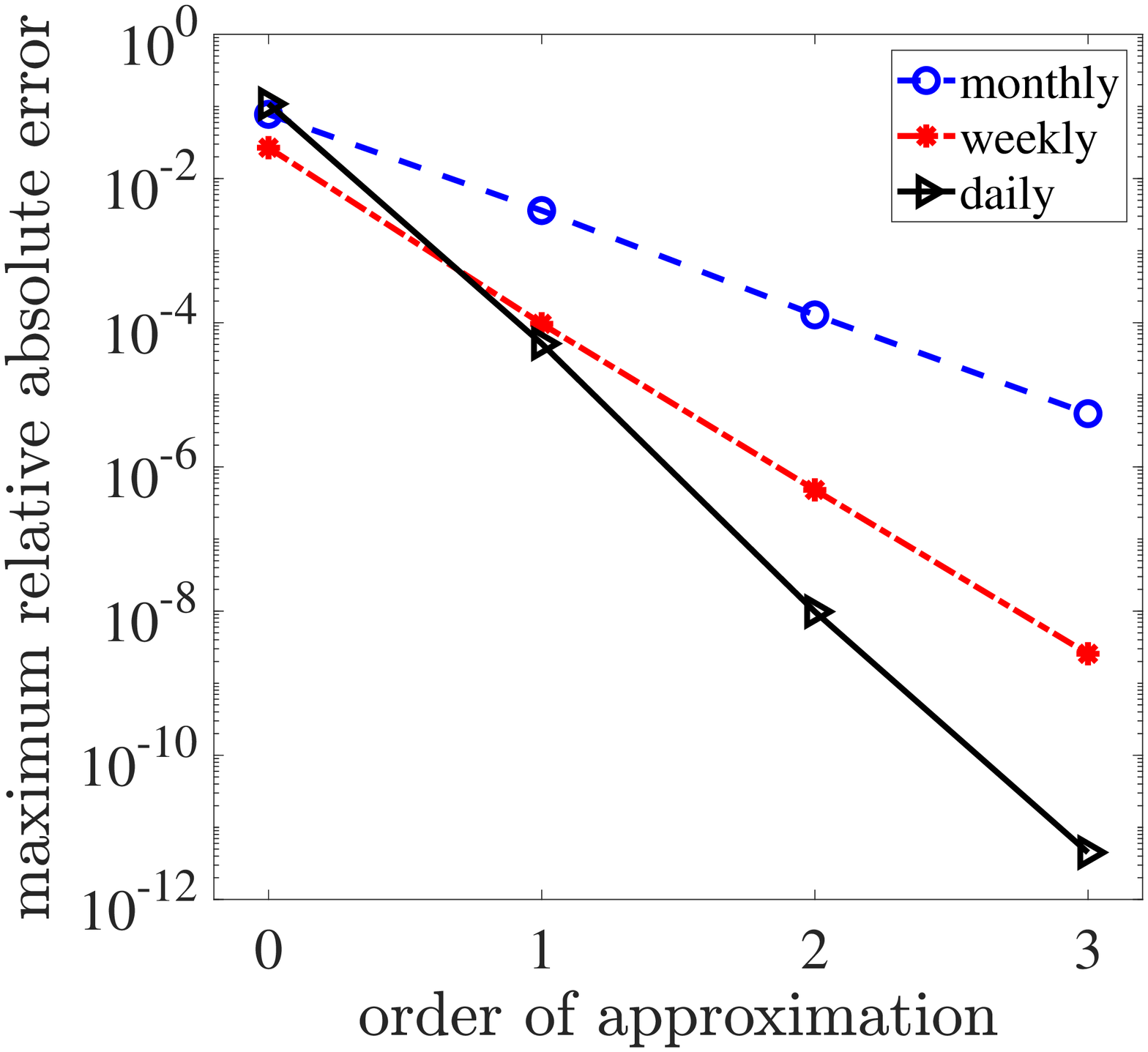}%
\\
(a) Pure jump OU model.
\end{center}}}
&
{\parbox[b]{2.1003in}{\begin{center}
\includegraphics[
height=1.7996in,
width=2.1003in
]%
{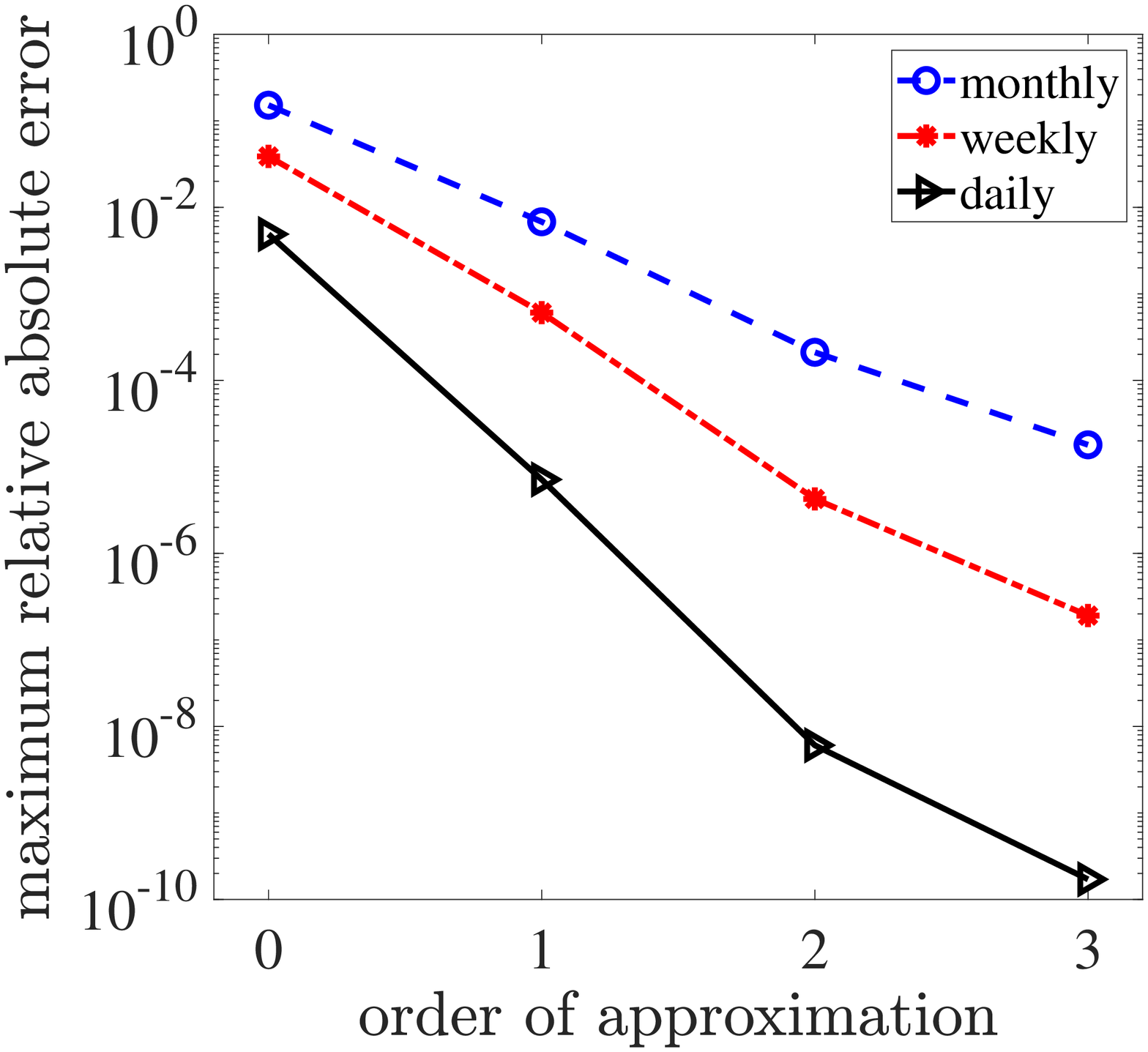}%
\\
(b) Constant diffusion model.
\end{center}}}
&
{\parbox[b]{2.1003in}{\begin{center}
\includegraphics[
height=1.7996in,
width=2.1003in
]%
{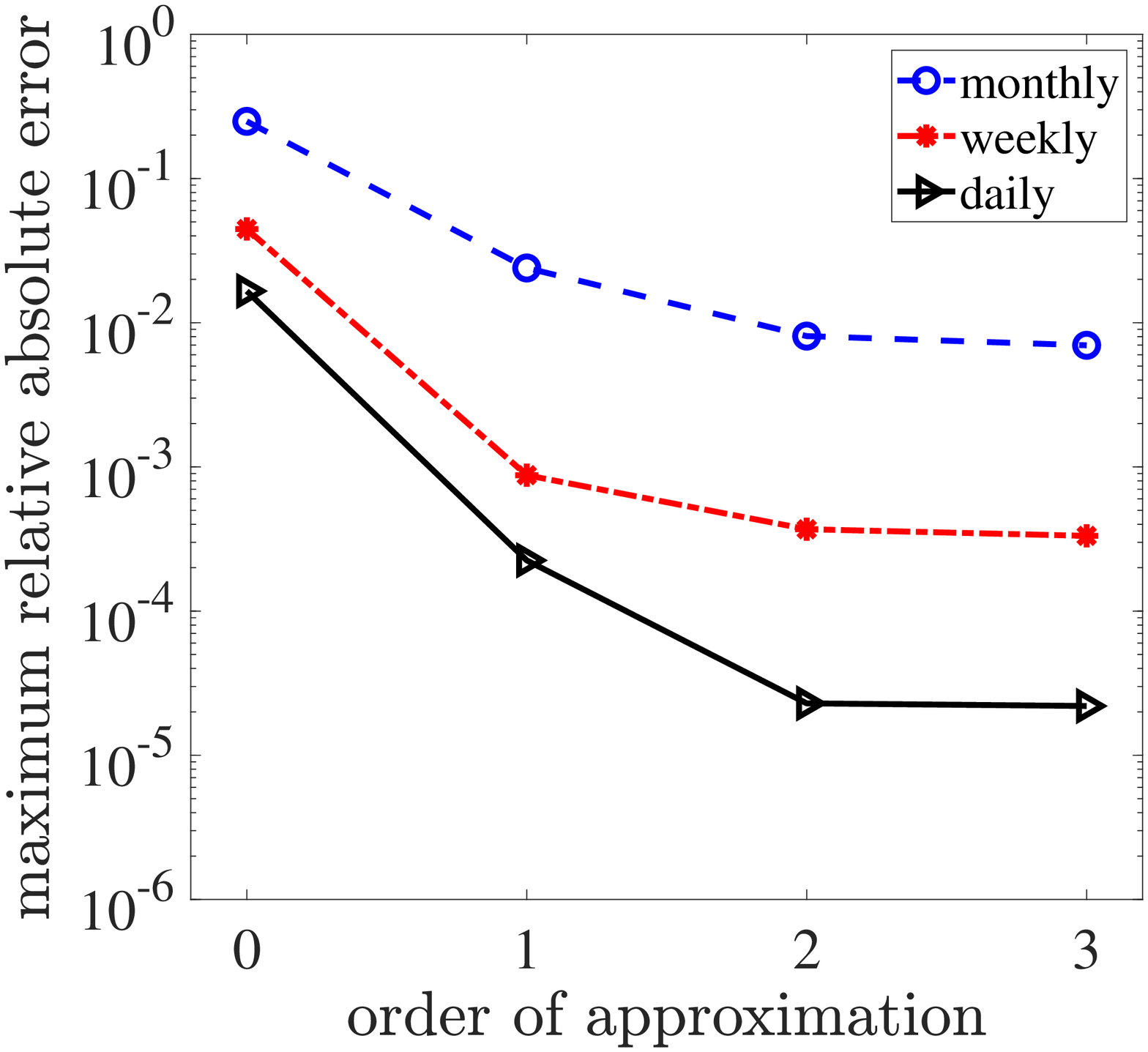}%
\\
(c) Square-root diffusion model.
\end{center}}}
\end{tabular}%
\caption
{Maximum relative absolute errors of density approximation for Models 1, 2 and 3 with orders $M=0,1,2,3$.}%
\label{fig:Max_relative_abso_err}%
\end{center}%
%

\renewcommand{\baselinestretch}{1.0}%
\noindent
\end{figure}%

In Figures \ref{fig:error Model 1}, \ref{fig:error Model 2} and
\ref{fig:error Model 3}, we plot the series of approximation errors defined by
(\ref{approximated error}) for the three models respectively. For each model,
we consider the case of $\Delta=1/52$ and denote by $e_{0},e_{1},e_{2},e_{3}$
the abbreviation of $e_{M}\left(  \left.  \Delta,x\right\vert x_{0}%
;\theta\right)  $ in (\ref{approximated error}) for $M=0,1,2,3$ respectively.
We observe from Figures \ref{fig:error Model 1} -- \ref{fig:error Model 3}
that the approximation errors decrease quickly and consistently as the order
of expansion increases.%

\begin{figure}[H]%
\begin{center}%
\begin{tabular}
[c]{cccc}%
{\parbox[b]{1.5003in}{\begin{center}
\includegraphics[
height=1.2994in,
width=1.5003in
]%
{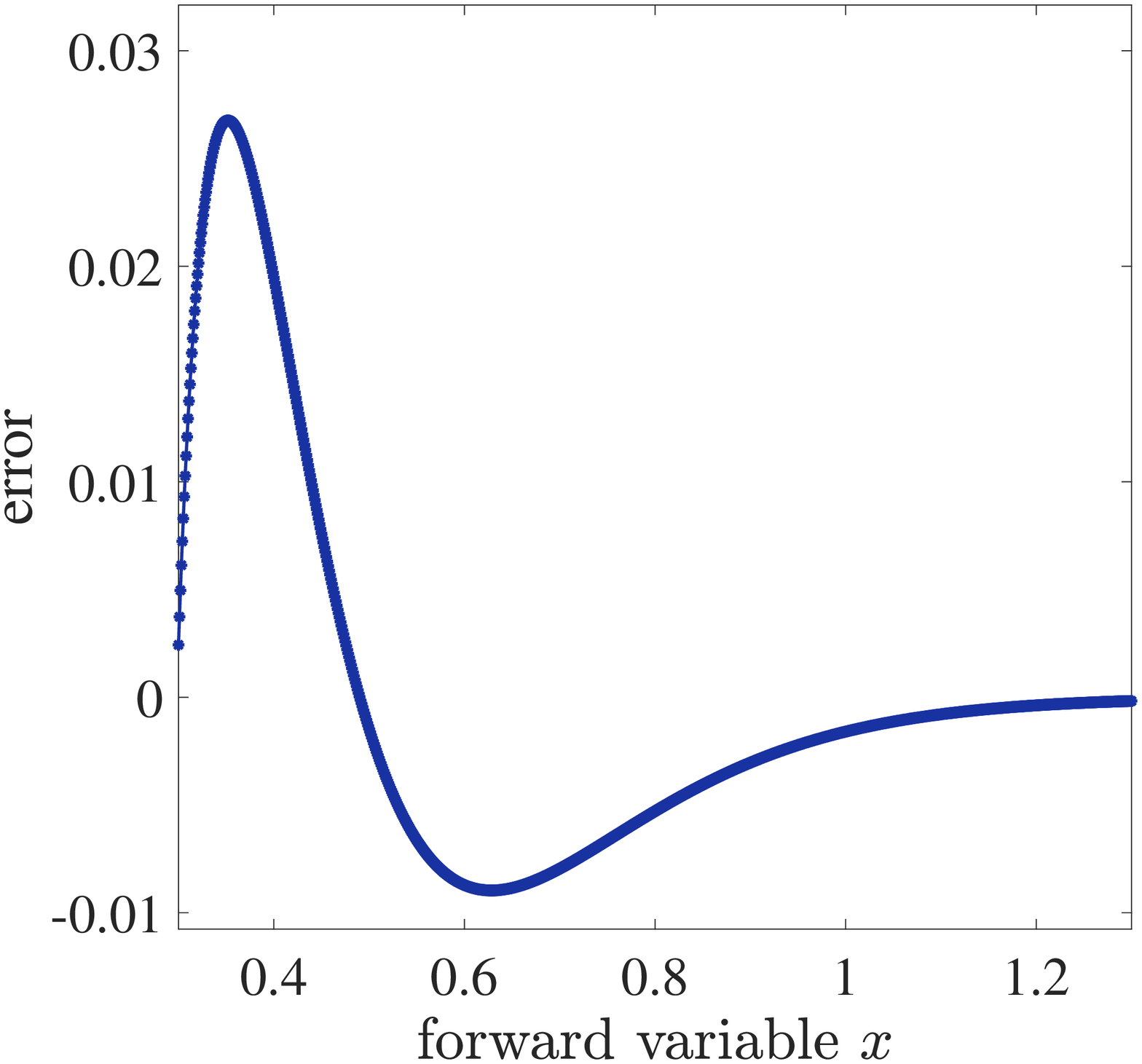}%
\\
(a) $e_{0}$.
\end{center}}}
&
{\parbox[b]{1.5003in}{\begin{center}
\includegraphics[
height=1.2994in,
width=1.5003in
]%
{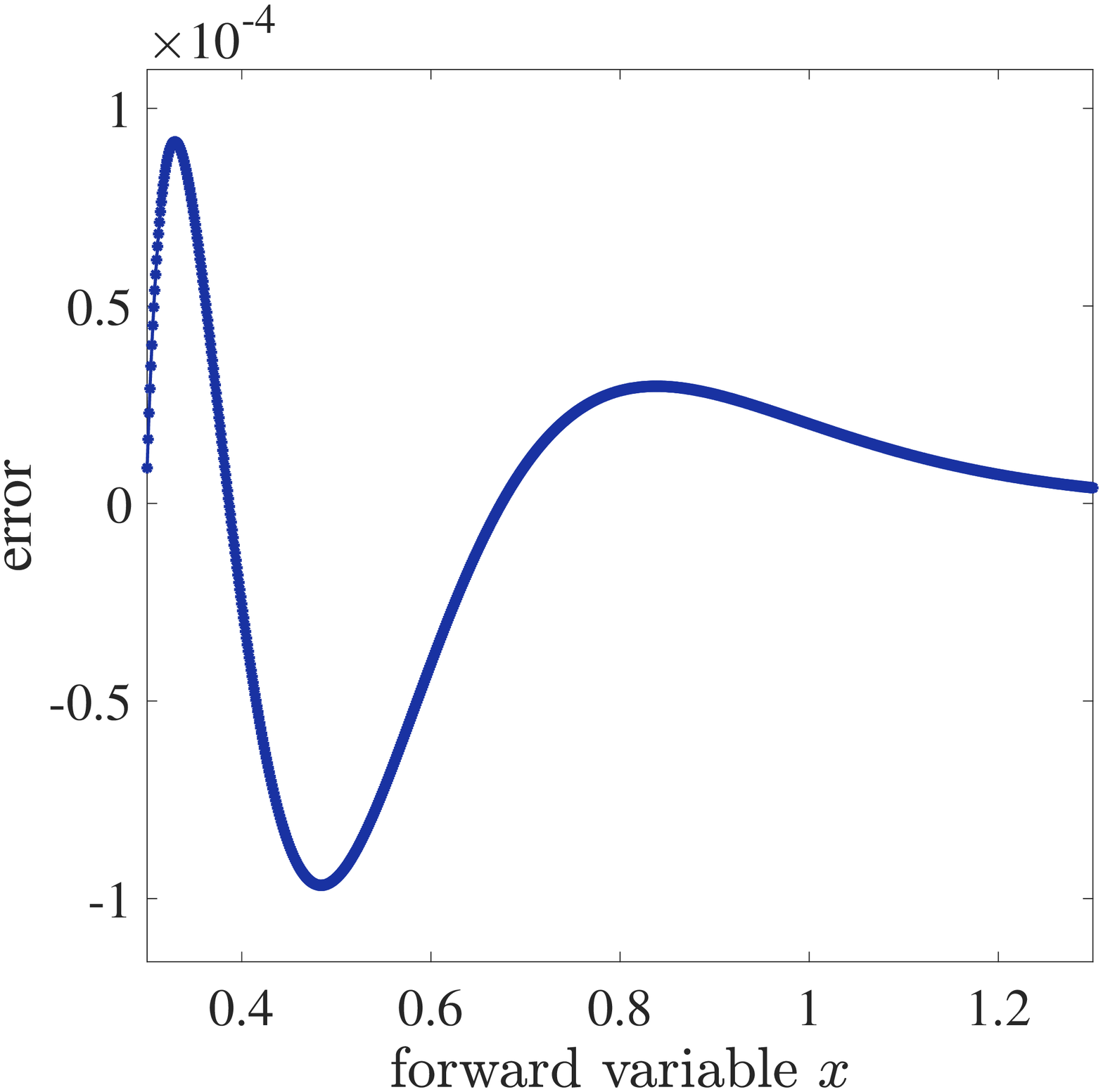}%
\\
(b) $e_{1}$.
\end{center}}}
&
{\parbox[b]{1.5003in}{\begin{center}
\includegraphics[
height=1.2994in,
width=1.5003in
]%
{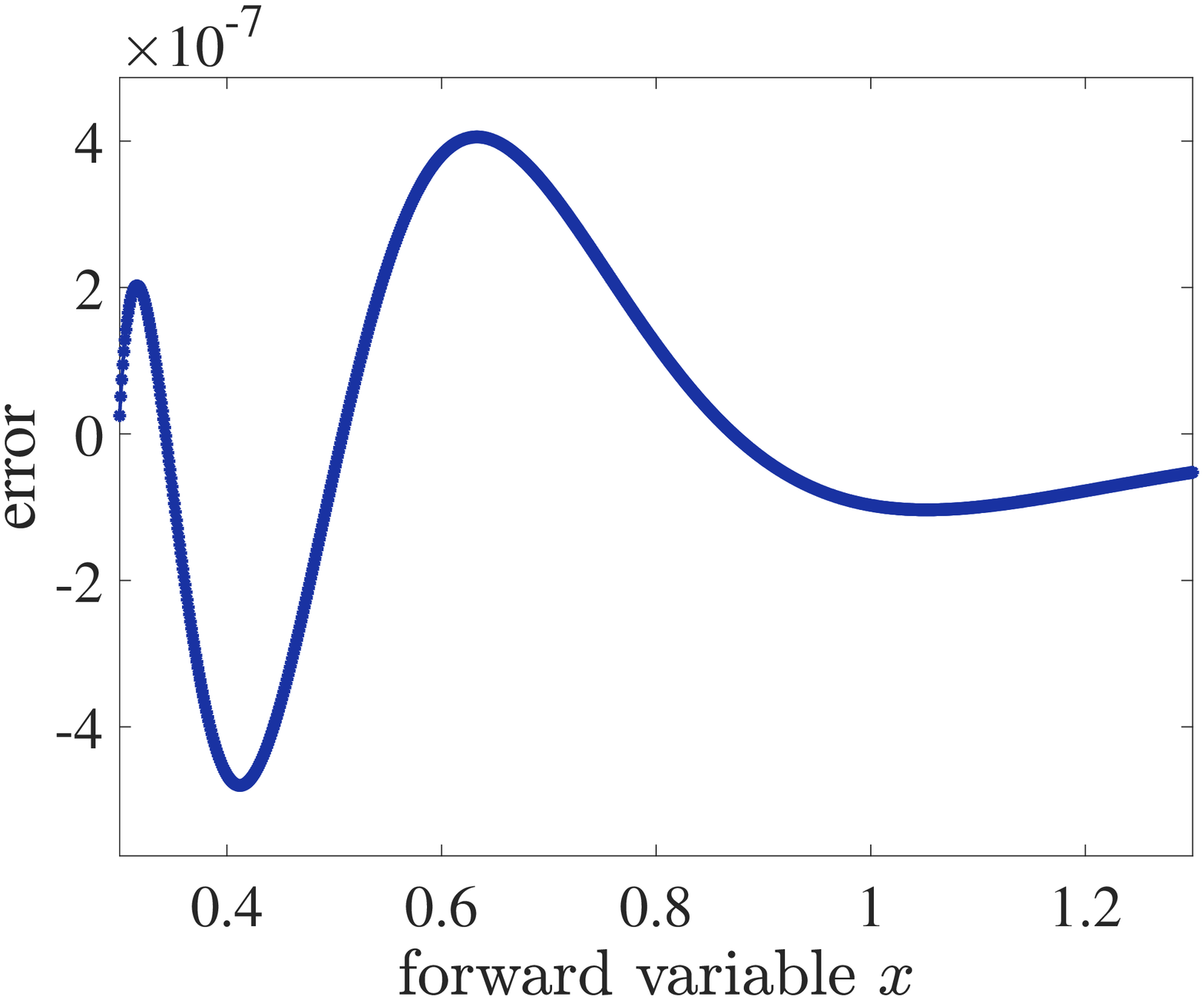}%
\\
(b) $e_{2}$.
\end{center}}}
&
{\parbox[b]{1.5003in}{\begin{center}
\includegraphics[
height=1.2994in,
width=1.5003in
]%
{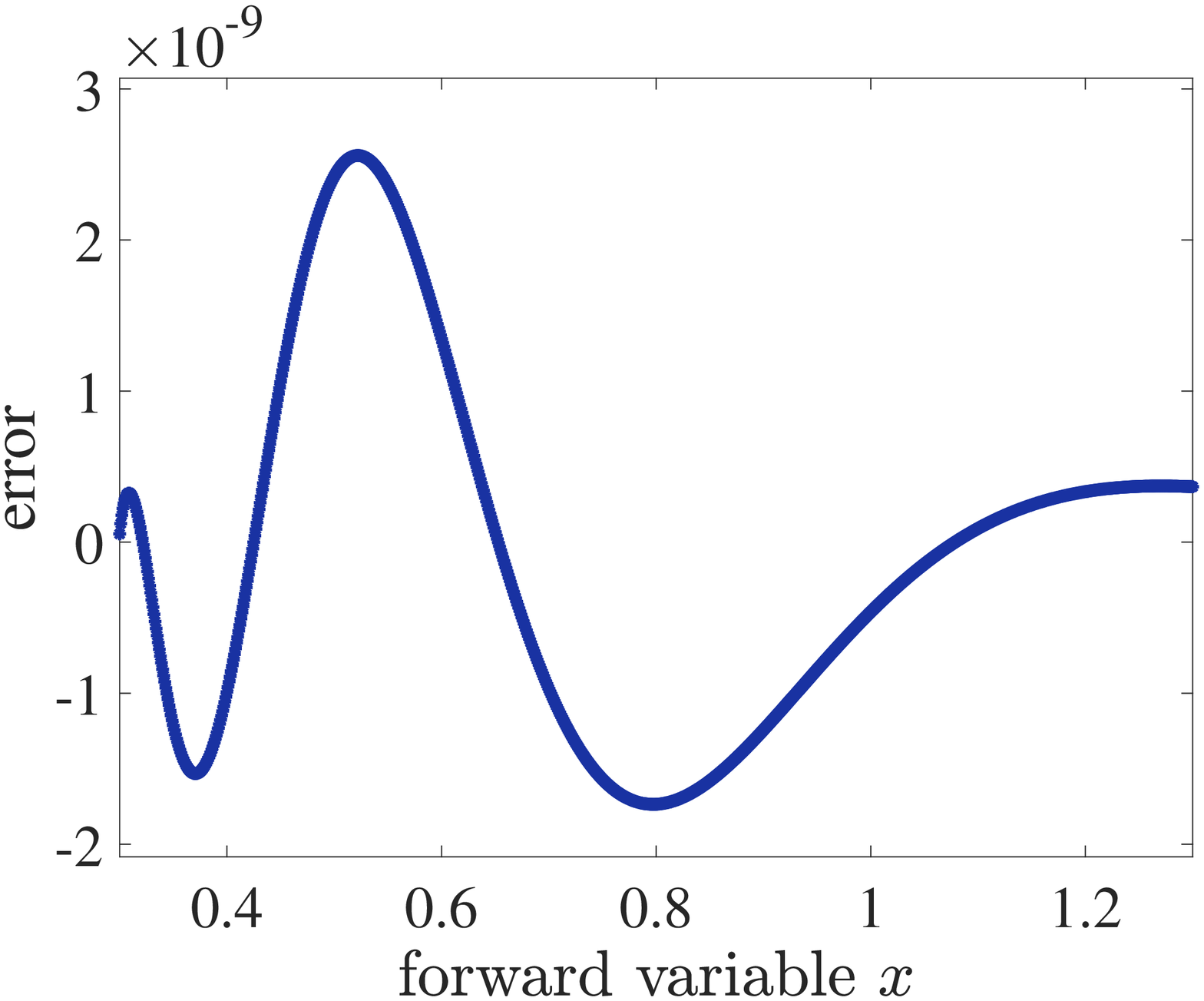}%
\\
(c) $e_{3}$.
\end{center}}}
\end{tabular}%
\caption
{Errors of density approximation for weekly monitoring frequency ($\Delta
=1/52$) in pure jump OU model, i.e., $e_{0},e_{1},e_{2},e_{3}%
$ as the approximation errors with respect to the expansion orders $M=0,1,2,3$ respectively.}%
\label{fig:error Model 1}%
\end{center}%
%

\renewcommand{\baselinestretch}{1.0}%
\noindent
\end{figure}%
%

\begin{figure}[H]%
\begin{center}%
\begin{tabular}
[c]{cccc}%
{\parbox[b]{1.5003in}{\begin{center}
\includegraphics[
height=1.2994in,
width=1.5003in
]%
{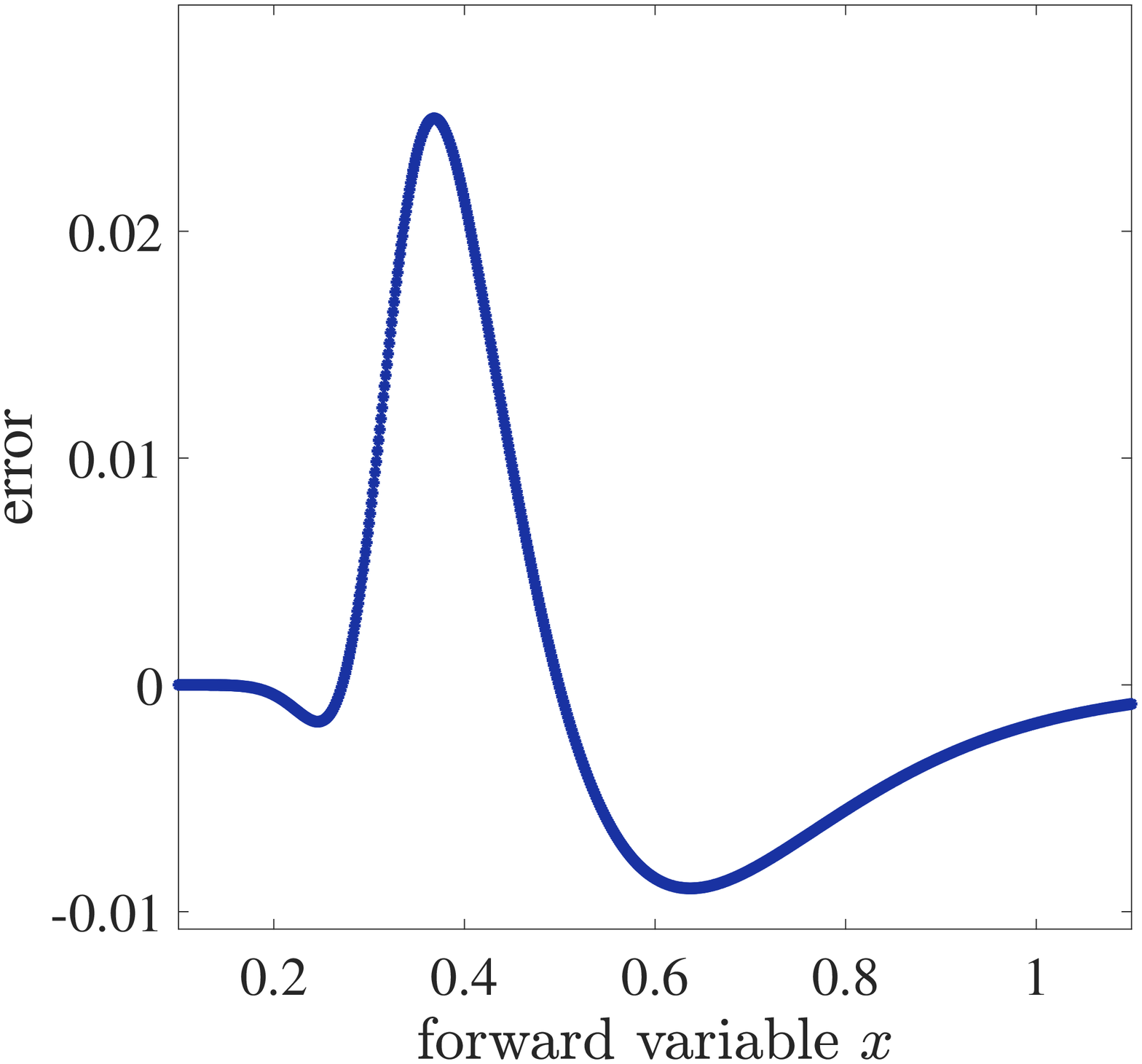}%
\\
(a) $e_{0}$.
\end{center}}}
&
{\parbox[b]{1.5003in}{\begin{center}
\includegraphics[
height=1.2994in,
width=1.5003in
]%
{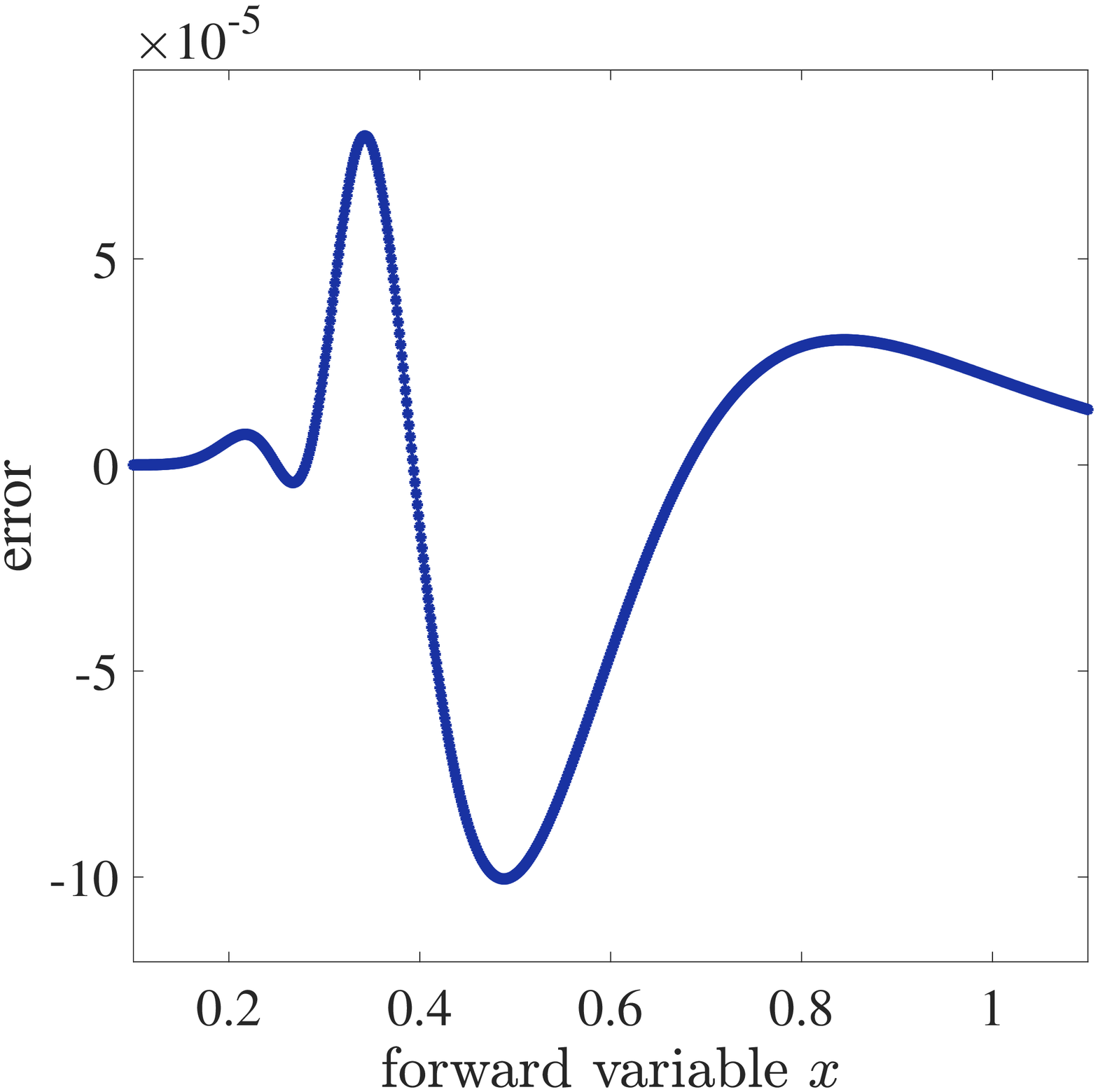}%
\\
(b) $e_{1}$.
\end{center}}}
&
{\parbox[b]{1.5003in}{\begin{center}
\includegraphics[
height=1.2994in,
width=1.5003in
]%
{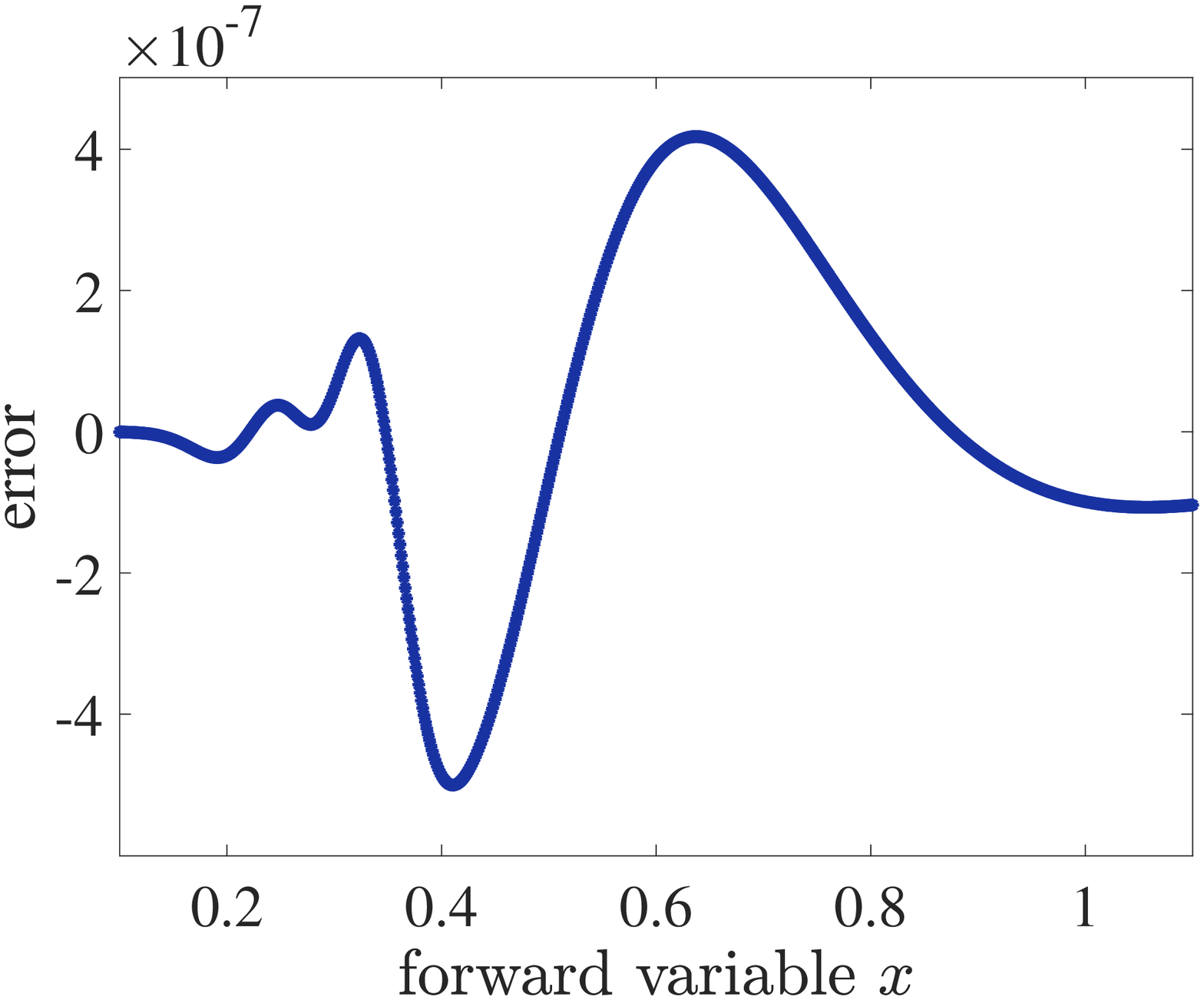}%
\\
(b) $e_{2}$.
\end{center}}}
&
{\parbox[b]{1.5003in}{\begin{center}
\includegraphics[
height=1.2994in,
width=1.5003in
]%
{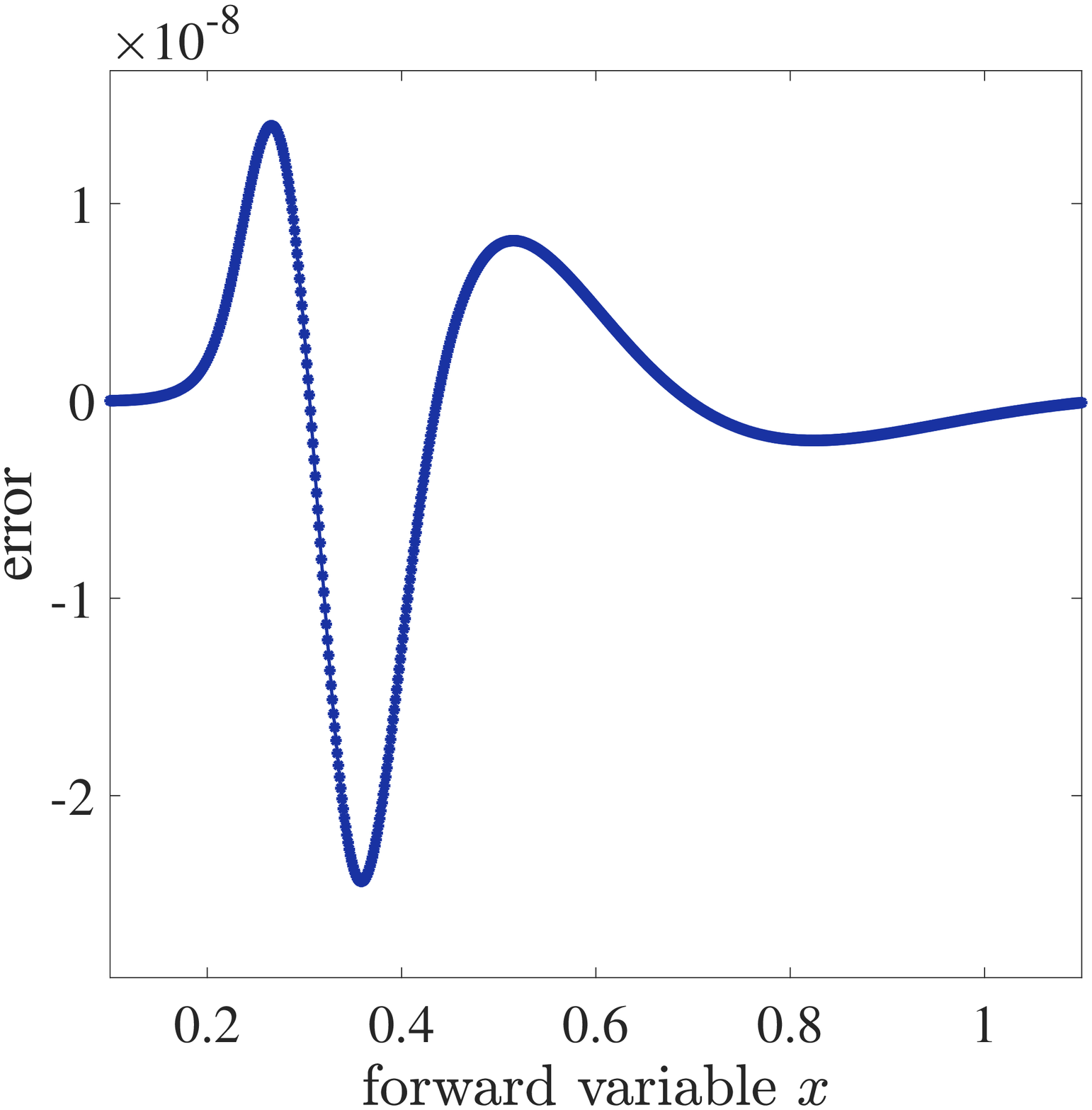}%
\\
(c) $e_{3}$.
\end{center}}}
\end{tabular}%
\caption
{Errors of density approximation for weekly monitoring frequency ($\Delta
=1/52$) in constant diffusion model, i.e., $e_{0},e_{1},e_{2},e_{3}%
$ as the approximation errors with respect to the expansion orders $M=0,1,2,3$ respectively.}%
\label{fig:error Model 2}%
\end{center}%
%

\renewcommand{\baselinestretch}{1.0}%
\noindent
\end{figure}%
%

\begin{figure}[H]%
\begin{center}%
\begin{tabular}
[c]{cccc}%
{\parbox[b]{1.5003in}{\begin{center}
\includegraphics[
height=1.2994in,
width=1.5003in
]%
{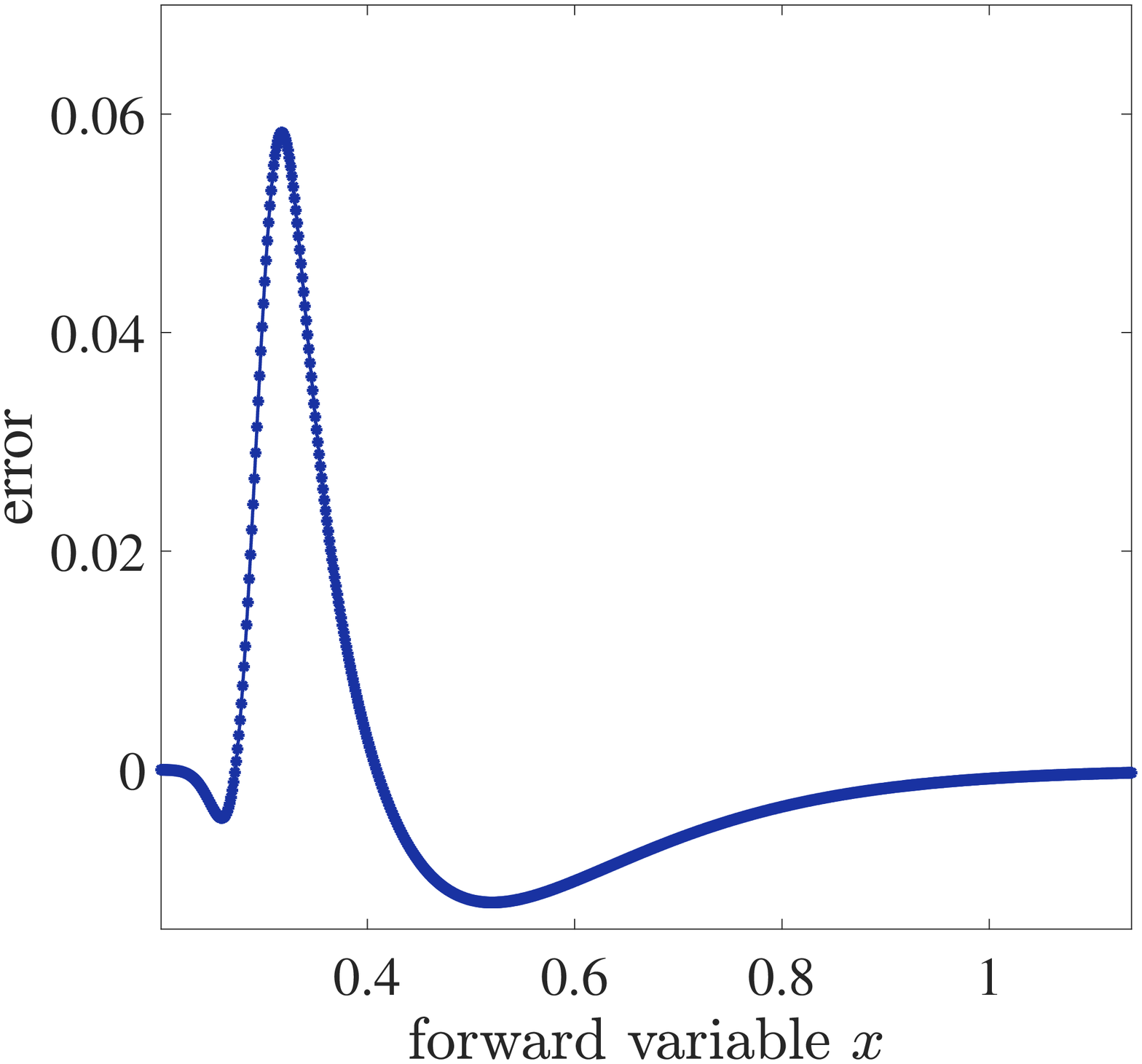}%
\\
(a) $e_{0}$.
\end{center}}}
&
{\parbox[b]{1.5003in}{\begin{center}
\includegraphics[
height=1.2994in,
width=1.5003in
]%
{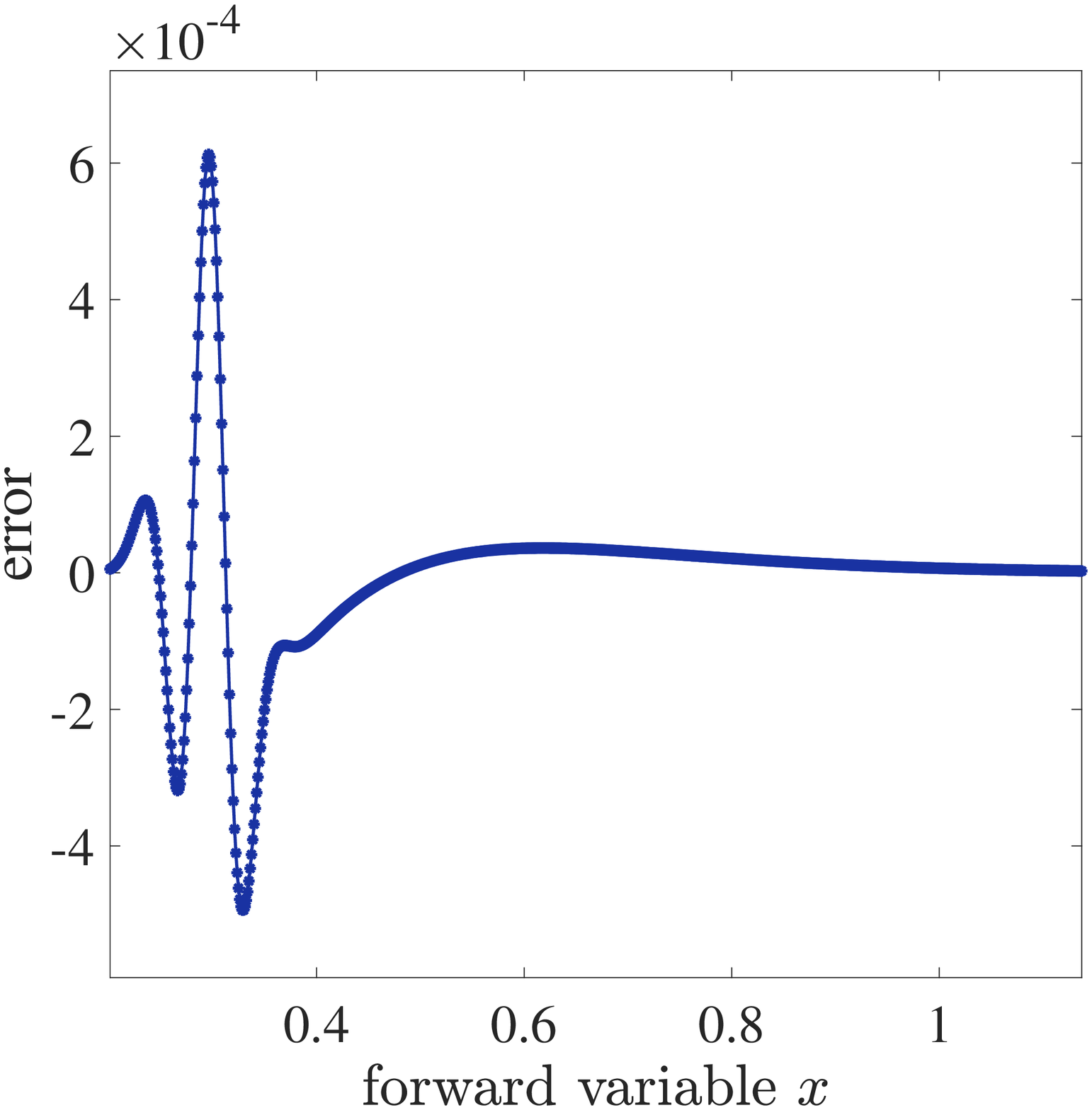}%
\\
(b) $e_{1}$.
\end{center}}}
&
{\parbox[b]{1.5003in}{\begin{center}
\includegraphics[
height=1.2994in,
width=1.5003in
]%
{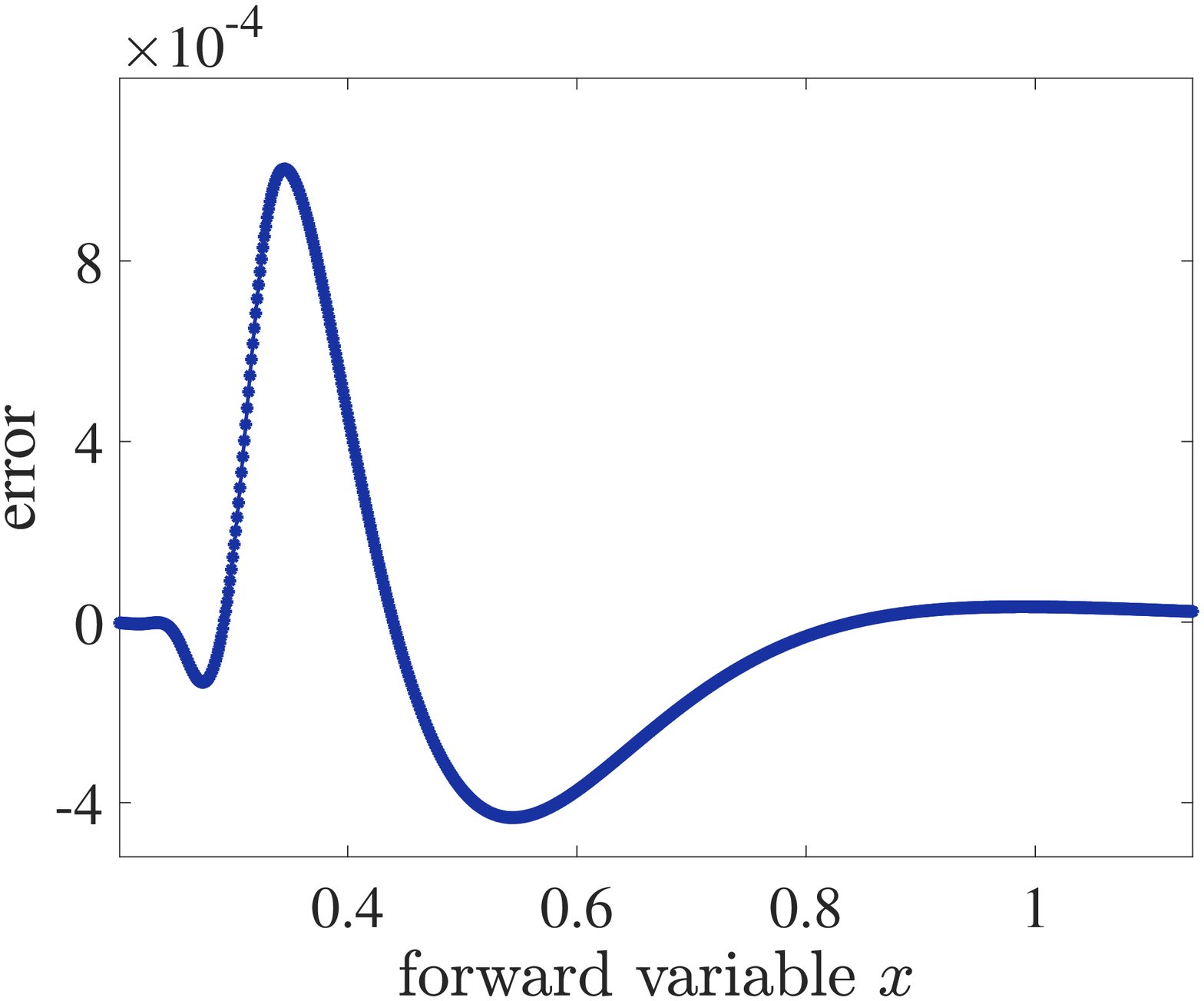}%
\\
(b) $e_{2}$.
\end{center}}}
&
{\parbox[b]{1.5003in}{\begin{center}
\includegraphics[
height=1.2994in,
width=1.5003in
]%
{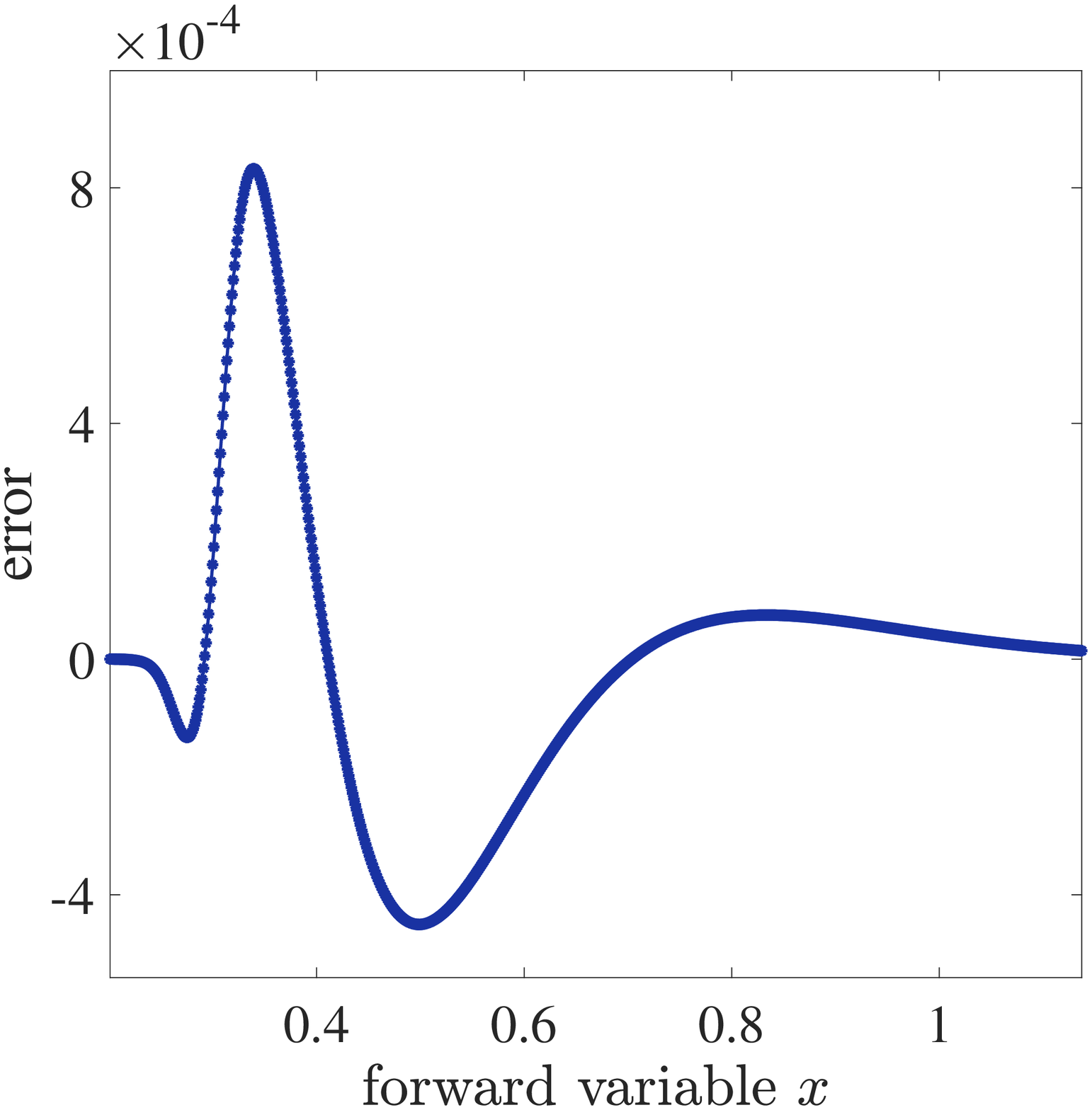}%
\\
(c) $e_{3}$.
\end{center}}}
\end{tabular}%
\caption
{Errors of density approximation for weekly monitoring frequency ($\Delta
=1/52$) in square-root diffusion model, i.e., $e_{0},e_{1},e_{2}%
,e_{3}%
$ as the approximation errors with respect to the expansion orders $M=0,1,2,3$ respectively.}%
\label{fig:error Model 3}%
\end{center}%
%

\renewcommand{\baselinestretch}{1.0}%
\noindent
\end{figure}%

The advantages of our asymptotic expansion method over the method of inverse
Fourier transform can be summarized as follows:

\begin{description}
\item[(1)] When the solution of SDE (\ref{model}) does not admit an explicit
expression of $X(t)$ or characteristic function $\phi\left(  t;\omega\right)
$, our method can still be used to approximate the transition density.

\item[(2)] The approximation errors decrease quickly as the approximation
order $M$ in (\ref{pX_M_expan}) increases, thus it suffices to use the first
several expansion terms for the approximation.

\item[(3)] Since the SDE driven by gamma process involves the fat-tail
characteristic, the characteristic function $\phi\left(  t;\omega\right)  $
for the pure jump SDE decreases slowly when $\omega\rightarrow+\infty$, which
induces heavy computation burden in the inverse Fourier transform method. In
contrast, our expansion terms for the pure jump SDE can achieve quick
convergence and be evaluated in a few seconds for any rational initial value
of $X_{0}$.
\end{description}

\section{Concluding remarks}

\label{sec:conclusion}

In this paper, we propose a closed-form asymptotic expansion to approximate
the transition density of the jump-diffusion SDE driven by the gamma process.
We employ three examples with known characteristic functions for numerical
illustrations and comparisons. Compared with the method of calculating the
transition density by inverse Fourier transform of the characteristic
function, our method is more efficient while achieving low approximation
errors. In terms of the applications in financial engineering, our
approximation method can be directly applied for option pricing and hedging to
obtain analytically tractable results, which is left for further study.

\subsection*{Acknowledgements}

Jiang and Yang's research was supported by the National Natural Science
Foundation of China (Grants No. 11671021).

\bibliographystyle{elsevier}
\bibliography{gammaexpansion}

\end{document}